\documentclass[10pt, journal]{IEEEtran}

\usepackage{graphics}
\usepackage{amsmath}
\usepackage{gensymb}
\usepackage{amssymb}
\usepackage{subfigure,graphicx}
\usepackage{epstopdf}
\usepackage{bm,color}
\usepackage{algorithm}
\usepackage{algpseudocode}
\usepackage{mathtools,cuted}
\usepackage{lipsum, color}
\usepackage{multicol}
\usepackage{multirow}
\usepackage{booktabs}
\usepackage{cases}
\usepackage{amsfonts}
\usepackage[colorlinks,linkcolor=red,anchorcolor=blue,citecolor=blue]{hyperref}

\newtheorem{definition}{Definition}

\newtheorem{rmk}{Remark}
\newtheorem{thm}{Theorem}
\newtheorem{lemma}{Lemma}
\newtheorem{assumption}{Assumption}
\newtheorem{proposition}{Proposition}

\newenvironment{proof}{\hspace{2ex}\emph{Proof}.\hspace{1ex}}{\hfill$\Box$}

\begin{document}

\title{Mobile Formation Coordination and Tracking Control for Multiple Non-holonomic Vehicles}

 \author{~Xiuhui~Peng,~Zhiyong~Sun,~\IEEEmembership{Member,~IEEE}, ~Kexin~Guo and~Zhiyong~Geng

 \thanks{This work is supported by the National Natural Science Foundation(NNSF)
 of China under Grant 61773024.}
 \thanks{Xiuhui Peng and Zhiyong Geng are with State Key Laboratory for Turbulence and Complex Systems,
 Department of Mechanics and Engineering Science, College of Engineering, Peking University, Beijing 100871,
 China. Email:xhpeng2014@pku.edu.cn, zygeng@pku.edu.cn.}
 \thanks{Zhiyong Sun is  with   Department of Automatic Control, Lund
 University, Sweden. Email: zhiyong.sun@control.lth.se, sun.zhiyong.cn@gmail.com }
 \thanks{Kexin Guo is with the School of Electrical and Electronic
 Engineering, Nanyang Technological University, 639798, Singapore. E-mail:
 guok0005@e.ntu.edu.sg.}
 }

\maketitle

\begin{abstract}
This paper addresses  forward motion control  for trajectory tracking and mobile  formation coordination   for a group of non-holonomic vehicles on $SE(2)$. Firstly, by  constructing  an intermediate attitude variable which involves  vehicles' position information and desired attitude,   the translational and rotational control inputs are designed in two stages to solve the  trajectory tracking problem.
Secondly, the coordination relationships of relative positions and headings are explored thoroughly for a group of non-holonomic vehicles to maintain  a mobile formation with rigid body motion constraints. We prove that, except for the cases of parallel formation and translational straight line formation,  a  mobile  formation with strict rigid-body motion can be achieved if and only if the ratios of linear speed to angular speed for each individual vehicle are constants. Motion properties for  mobile formation with weak rigid-body motion are also demonstrated.
Thereafter, based on the proposed trajectory tracking approach, a distributed mobile   formation control law is designed under a directed tree graph.
The performance of the proposed  controllers is validated by both numerical simulations and experiments.
\end{abstract}


{}

\definecolor{limegreen}{rgb}{0.2, 0.8, 0.2}
\definecolor{forestgreen}{rgb}{0.13, 0.55, 0.13}
\definecolor{greenhtml}{rgb}{0.0, 0.5, 0.0}

\section{Introduction}


\IEEEPARstart{R}{elated} research topics of multi-vehicles coordination  control include  dynamics of vehicles (e.g.,   integrator-based dynamics, linear or nonlinear systems,   Euler-Lagrange (EL) systems, etc) \cite{rencao}, communication  modes (e.g., undirected or directed graphs, fixed or switching communication topology, communication delays, etc) \cite{dongxiwang2016}, and various coordination tasks (e.g., consensus, flocking, formation, rigid shape swarming, etc) \cite{oh2015}. In practice, vehicles often suffer from under-actuated constraints, such as the well-known non-holonomic constraints in the planar motion of unmanned vehicles \cite{mag2018}, which should be taken into consideration in coordination control design. In this paper, by taking into account these motion constraints, we focus on  the mobile formation coordination for multiple non-holonomic vehicles.



First of all, we study coordination control with guaranteed forward motion to achieve vehicle trajectory tracking. In practice, the kinematics of non-holonomic vehicles, such as ground vehicles or  unmanned aerial vehicles (UAVs), are often described by a unicycle model. Therefore, besides the  non-holonomic motion constraint,   a speed   constraint should be taken into account since the vehicles only have forward motion (e.g., positive airspeeds for fixed-wing UAVs  {\cite{weitz2008,Stipanovic2004,duan2009}}). One similar scenario, for example, is that the backward motion is not allowed when a group of unicycle-type vehicles  pass though a narrow path. This raises an interesting control problem that a positive forward motion  must be maintained all the time. To solve this problem, some papers, {e.g.,  \cite{sun,qinjiahu,ma}} assume that the linear speed is a positive constant and mainly focus on the design of  angular speed input. However, there are limited results that take both linear speed and angular speed constraints into the design consideration. In {\cite{weitz2008,Stipanovic2004}}, the forward motion is solved by taking acceleration as an auxiliary design variable (i.e., a dynamic controller), and a model predictive control (MPC) approach which demands high online computational cost is presented in {\cite{duan2009}}. From a driver's perspective, the heading of the vehicle needs to be adjusted simultaneously according to  the position error to a target vehicle so that the tracking task can be achieved, and this control framework has been adopted in tracking control of under-actuated quadrotors {\cite{ding2016,zouyao2}}. In this paper, we extend this idea to design the forward motion controller for non-holonomic vehicles. The novelty of the proposed control law is that it can guarantee a forward motion all the time,  which can be applied to coordinate  UAV-type systems with positive forward speeds.


Beyond the single-agent control, formation control as a typical multi-agent application has been investigated widely, which aims to control multiple vehicles to form and maintain a prescribed geometric shape \cite{oh2015}. Many approaches, such as behavior-based approach {\cite{lawton2003}}, virtual structure \cite{do2007},  leader-follower structure  {\cite{das2002}}, potential field method {\cite{cheah}},  consensus-based approach {\cite{dongxiwang2}}, and MPC approach {\cite{sunzhongqi,sunzhongqi2}},   are proposed to tackle the formation control problem for non-holonomic vehicles. Besides the formation stabilization control approaches, the research of mobile formation maneuvers is also an important issue, especially for the mobile formation coordination with rigid-body motion constraints.

Based on   position errors in the global frame,  the formation tasks for multiple non-holonomic vehicles have been studied in \cite{yuxiao,liutengfei1,liutengfei2,miaozhiqiang}. In those papers, the formation control schemes  only consider formation coordination control with translational  motions (i.e., the formation shape only has translational motions while rotations are not permitted).  Based on the leader-follower structure, the desired formation shapes in {\cite{das2002,liangxinwu}} are defined in the leader's coordinate frame. For example,  the formation shape stabilization problem where the follower maintains a desired distance and orientation with respect to its leader's coordinate frame  is studied in {\cite{das2002}}, and the formation control problem that the follower maintains a desired relative position with respect to its leader's coordinate frame  is studied in {\cite{liangxinwu}}. As an extension of the leader-follower scheme, the formation shapes in {\cite{Consolini2008,Morbiditac,liangxinwu2}} are specified in the follower's coordinate frame. For example,   a desired distance and orientation of the leader's barycenter with respect to the follower's coordinate frame is used in {\cite{Consolini2008,Morbiditac}}, and the formation shape in {\cite{liangxinwu2}} is specified by  a desired relative position of the leader's barycenter with respect to the follower's coordinate frame. Based on the graph  rigidity theory, the paper {\cite{milad}} has proposed  rigid-graph-based formation control laws to achieve and maintain a rigid formation shape.
In addition, the moving target circular formation task has been studied in {\cite{Arranz}}.
However, those formation maneuver controllers cannot maintain the mobile formation shape with a weak/strict rigid-body motion (definitions will be clear in the context).  In the formation control field, certain important applications, such as rigid formation pattern and aircraft military maneuver,  among others, often require to maintain fixed relative positions for all vehicles with respect to  one common vehicle (namely, a mobile formation with weak rigid-body motion).  Some particular application scenarios, including aircraft refueling or satellite docking, require  a fixed relative position between any two vehicles (namely, a mobile  formation with strict rigid-body motion).   By taking into account the non-holonomic constraint,   formation control and motion coordination  with weak/strict rigid-body motion for  multiple non-holonomic vehicles becomes even more challenging.  To our best knowledge, the condition of  mobile formation with weak/strict rigid-body motion for  non-holonomic vehicles has not ever been discussed before, which will be one of the main focuses in this paper.



Formation control with fixed or rigid shapes for multiple single- or double-integrator agents (i.e., point-mass type models) has been surveyed in \cite{Anderson2008}, wherein most control laws are constructed  with position errors and desired distances among  agents. Nevertheless, compared with integrator-based vehicle models, the research on mobile formation with weak/strict rigid-body motion for vehicles which are modeled by nonlinear manifolds becomes more meaningful.  The fixed distances and relative configurations can be straightforwardly regulated   for fully-actuated planar vehicles \cite{dong2013,liuyongfang}. However, it still remains  unclear to explore  the relationship among  fixed distances, fixed relative positions, headings, and speed constraints   of multiple non-holonomic vehicles while maintaining a mobile formation with weak/strict rigid-body motion.  In the light of the body-fixed frame and inertial frame, several interesting properties on mobile  formation behavior  for   multiple non-holonomic vehicles are first explored in this paper. Different from  the rigid formation control with graph rigidity condition in \cite{Anderson2008},  the uniqueness of mobile  formation with weak/strict rigid-body motion can be achieved by certain specified formation tasks involving relative positions and headings in this paper.   Thereby, a distributed mobile formation control law for coordinating multiple vehicles with strict rigid-body motion is designed under a directed tree graph.

To summarize, the main contributions of this paper are listed as follows:

\begin{enumerate}
\item[1.] A novel forward motion controller  is proposed  to realize tracking control of SE(2) non-holonomic vehicles. By coupling the position error and desired attitude information,  an intermediate attitude is presented for non-holonomic vehicles to achieve forward motion control for  trajectory tracking to a leader vehicle. The   proposed control inputs  ensure a positive forward motion all the time. In addition, the saturation of inputs is also guaranteed.  The proposed results can be applied to not only  unicycle-type ground vehicles, but also fixed-wing UAVs flying in the planform.

\item[2.] The motion properties of relative positions and headings are  explored thoroughly  for a group of mobile non-holonomic vehicles maintaining a target formation shape. The proposed adjoint orbit and its properties are presented in the first two propositions in this paper. To ensure that any two vehicles have a fixed  relative position in a mobile formation, we prove that the ratios of linear speed to angular speed for each individual vehicle  have to be constants  except for the cases of parallel formation and translational straight line formation. Motion properties for  mobile formation with weak rigid-body motion are also demonstrated. To our best knowledge, it is the first time that such  necessary and sufficient conditions of mobile formation with weak/strict rigid-body motion for non-holonomic vehicles are studied.

\item[3.]  The control inputs for maintaining a mobile formation with weak/strict rigid-body motion for each vehicle are provided in this paper.  Based on the proposed
 tracking control law and mobile formation coordination theory,  a fully distributed mobile
 formation control law is designed to form and maintain a mobile formation with strict rigid-body motion.
\end{enumerate}

The remainder of this paper is structured as follows. The notations and problem  formulation are introduced in Section~\ref{sec:prob}. The forward motion controller and its stability analysis are presented in Section \ref{sec:twoStage}. Section \ref{sec:formation} introduces  mobile   formation coordination of multiple non-holonomic vehicles. Simulation and experiment results are shown in Section \ref{sec:sml}. Finally, conclusions are drawn in Section \ref{sec:cls}.

\section{Background and Problem formulation} \label{sec:prob}

\subsection{Notations}

Notations and concepts in this paper are fairly standard.  $\|x\|$ denotes the Euclidean norm of a vector $x$.
The symbol $x\in \mathbb{S}^{1}$ represents a vector $x\in \mathbb{R}^{2}$ with unit Euclidean norm.
Let $I_{p}$ denote the identity matrix of dimension $p\times p$.  Let ${e_{1},e_{2}}$ denote the natural bases of $\mathbb{R}^{2}$. Let the principal axes of the rigid body define a body-fixed
reference frame attached to the vehicle's center of mass and be denoted by $\mathcal{F}_{B}$. Let $\mathcal{F}_{I}$
be the inertial frame.  Let $\sigma(\cdot)$ define a saturation function which satisfies $|\sigma(x)|<1$, $\sigma(0)=0$ and $x\sigma(x)>0$ for all $x\neq0$, where $x\in \mathbb{R}^{n}$. Examples of the function $\sigma(x)$ include $\tanh(x)$ and $\frac{x}{\sqrt{1+x^{2}}}$.

In this paper, we use $SE(2)$ to describe the configuration space of a planar vehicle. Special Orthogonal group is denoted as $SO(2):=\{R\in \mathbb{R}^{2\times2}:R^{T}R=I_{2}, \mathrm{det}(R)=1\}$, and $p\in \mathbb{R}^{2}$ represents the position. The group element $g\in SE(2)$ is denoted  by
$$g=\left[
               \begin{matrix}
              R & p \\
               0 &  1
               \end{matrix}
               \right]=\left[
               \begin{matrix}
              \cos\theta & -\sin\theta & x \\
               \sin\theta & \cos\theta & y \\
               0 & 0 & 1
               \end{matrix}
               \right]$$
where a rotation matrix $R \in SO(2)$ describes the orientation from $\mathcal{F}_{B}$ to $\mathcal{F}_{I}$. The Lie algebra $\mathfrak{se}(2)$  denotes the velocity of a vehicle  in $\mathcal{F}_{B}$. A Lie algebra element $\hat{\xi} \in \mathfrak{se}(2)$ is denoted as
$$\hat{\xi}=\left[
               \begin{matrix}
              \hat{\omega} & \nu \\
               0 &  0
               \end{matrix}
               \right]=\left[
               \begin{matrix}
              0  & -\omega & \nu_{x} \\
               \omega & 0 & \nu_{y} \\
               0 & 0 & 0
               \end{matrix}
               \right]$$
where $\xi=[\omega, \nu_{x}, \nu_{y}]^{T} \in \mathbb{R}^{3}$, $\nu \in \mathbb{R}^{2}$ represents the  translational
speed, $\omega \in \mathbb{R}$ is the angular speed and  $\hat{\omega}\in \mathfrak{so}(2)$. The hat operator $(\cdot)^{\wedge}: \mathbb{R} \rightarrow \mathfrak{so}(2)$ is a linear map, where $\mathfrak{so}(2):=\{\hat{x}\in \mathbb{R}^{2\times 2} |\hat{x}^{T}=-\hat{x}\}$.   The inverse operator  to the hat operator $(\cdot)^{\wedge}$ is the vee operator $(\cdot)^{\vee}: \mathfrak{so}(2)\rightarrow \mathbb{R}$. For all $g\in SE(2)$, $A,B \in \mathfrak{se}(2)$, the adjoint map $\mathrm{Ad}_{g}$ is denoted as $\mathrm{Ad}_{g}A=gAg^{-1}$.

Without side slipping, the velocity $\xi$ of vehicles should satisfy $\nu_{y}=0$, represented by the non-holonomic constraint $\dot{x}\cos\theta-\dot{y}\sin\theta=0$.
 In this paper, we let $v \triangleq \nu_{x}$ denote the linear speed of the non-holonomic vehicles.


\subsection{Problem Description}

Consider $n+1$ non-holonomic vehicles labeled with $i=0,1,\cdots,n$,  in which the equation of motion for vehicle $i$ is modeled as
\begin{subequations}  \label{model}
\begin{align} \label{modelp}
\dot{p}_{i}  =& v_{i}R_{i}e_{1} \\ \label{modelR}
\dot{R}_{i}=& R_{i}\hat{\omega}_{i}
\end{align}
\end{subequations}
where the subscript $i$ represents vehicle $i$.
The node $0$ represents the leader and other node indexes  are followers. We give the following assumption on the desired speeds $v_{0}$ and $\omega_{0}$ for a leader vehicle.

\begin{assumption} \label{desiredw}
The desired speeds $v_{0}$, $\omega_{0}$  are assumed to be bounded for all time, and $v_{0}$ is assumed to be positive, i.e., $v_{0}>0$ all the time.
\end{assumption}

For a non-holonomic unicycle-type ground vehicle, the linear speed $v_{i}$ can be positive, negative or zero. However, in practice such as fixed-wing UAVs  with zero or  negligible  speed in the direction  perpendicular to vehicles' headings,   a persistent forward motion with a positive linear speed $v_{i}$ should be guaranteed. We consider this motion condition as a constraint in the tracking and formation control design.

 The main objective of this paper is to solve two coordination problems for unicycle-type  vehicles  subject to non-holonomic constraint. The first one is to design an appropriate  control law to realize a forward motion control  that achieves  trajectory tracking  of non-holonomic vehicles with $v_{i}>0$ all the time. The proposed forward motion control  approach can be applied to the trajectory tracking, formation  guidance, and other applications involving  UAVs with a positive forward speed.   The second control problem  is to thoroughly explore motion properties of relative positions and headings for a group of non-holonomic vehicles to maintain a mobile  formation  with certain motion constraints, and then design a distributed coordination control law  to form and maintain  a mobile  formation with strict rigid-body motion for multiple non-holonomic vehicles.

\section{Two-vehicles case: forward motion tracking control design} \label{sec:twoStage}

In this section, we focus on designing forward motion control laws that ensure a vehicle with index $1$ to track the leader with index $0$. First,  control inputs are proposed to solve the forward motion control  problem  for non-holonomic vehicles. Then, stability analysis of the closed-loop system is provided.

\subsection{Control input  design}
In the following, the control inputs $v_{1}$ and $\omega_{1}$ are designed, which consist of three steps.

\emph{Firstly,} the translational control input $v_{1}$ is designed. Let $p_{01}=p_{1}-p_{0}$ denote the position error between vehicle $1$ and leader $0$.
Then, the  virtual control input  and linear speed input are given by
\begin{equation}  \label{trackui}
u_{1}=-k_{1}\sigma(p_{01})+v_{0}R_{0}e_{1}
 \end{equation}
 \begin{equation} \label{trackvi}
 v_{1}=\|u_{1}\|
 \end{equation}
where $u_{1}\in \mathbb{R}^{2}$ is a virtual control input vector and $k_{1}>0$ is a constant control gain.

\emph{Secondly,}  an intermediate attitude   $\mathcal{R}_{0} \in SO(2)$ is constructed, which is given by
\begin{equation} \label{R01}
\mathcal{R}_{0}=[r_{0}^{1},r_{0}^{2}]\in SO(2)
\end{equation}
with the vectors  defined by
\begin{equation*} \label{b1i}
r_{0}^{1}=\frac{u_{1}}{\|u_{1}\|} \in \mathbb{S}^{1}, \quad r_{0}^{2}=\left[
               \begin{matrix}
               -r_{0}^{1}(2,1)\\
               r_{0}^{1}(1,1)
               \end{matrix}  \right]  \in \mathbb{S}^{1}
\end{equation*}
The kinematics of the intermediate attitude $\mathcal{R}_{0}$ satisfies the equation $\dot{\mathcal{R}}_{0}=\mathcal{R}_{0}\hat{\varpi}_{0}$, and the angular speed $\varpi_{0}$ is derived as
\begin{equation}
\varpi_{0}=(\mathcal{R}_{0}^{T}\dot{\mathcal{R}}_{0})^{\vee}
\end{equation}

\emph{Finally,} the rotational control input $\omega_{1}$ is designed. In our proposed approach, the attitude of vehicle $1$ is required to track the intermediate attitude $\mathcal{R}_{0}$. Let $R_{01}=\mathcal{R}_{0}^{T}R_{1}\in SO(2)$ be the rotation error between the attitude of vehicle $1$ and the intermediate attitude $\mathcal{R}_{0}$. Thus, the angular speed input $\omega_{1}$ is designed by
\begin{equation} \label{trackwi}
\omega_{1}=-k_{2}\sigma((R_{01}-R_{01}^{T})^{\vee})+\varpi_{0}
\end{equation}
where $k_{2}>0$ is a constant control gain.
\begin{rmk}
From the kinematics equation $(\ref{model})$, the evolution of attitude $R_{i}$ is controlled merely by a self-governed angular speed input $\omega_{i}$. Nevertheless, the evolution of position $p_{i}$ depends on both the heading $R_{i}$ and the linear speed input $v_{i}$. The intuition of the control input design $v_1$ and $\omega_1$ comes from  the scenario of a manned vehicle tracking another vehicle, wherein a driver constantly adjusts the heading of the vehicle according to the visible  position error, as illustrated in  Fig.~$\ref{direction}$. It can be seen that the direction of the position error vector $p_{0}-p_{1}$ points from the mass center of  vehicle $1$ to the mass center of leader $0$. A reasonable manner is to drive the heading $R_{1}$ of vehicle $1$ to track the direction vector $p_{0}-p_{1}$. Thus, by enhancing the couplings between the position and attitude variables, we propose a framework for designing  translational controller and attitude controller in two stages as shown in Fig.~$\ref{frame}$. In this framework, an intermediate attitude $\mathcal{R}_{0}$ in $(\ref{R01})$ is constructed by the position error vector ($p_{0}-p_{1}$) and the desired attitude $R_{0}$, and therefore the actual heading angle $R_{1}$ tracks the intermediate direction angle $\mathcal{R}_{0}$. Once the position error converges to zero, the intermediate rotation matrix $\mathcal{R}_{0}$ will be the same to the desired rotation matrix $R_{0}$ since the rotation matrix $R \in SO(2)$ is a one-parameter subgroup \cite{Bullo2005}.
 \end{rmk}

\setlength{\abovecaptionskip}{0.cm}
\setlength{\belowcaptionskip}{-0.cm}
	\begin{figure}[!htb]
	\centering
	{		
		\subfigure[]
		{\label{direction}\includegraphics[width=0.4\linewidth]{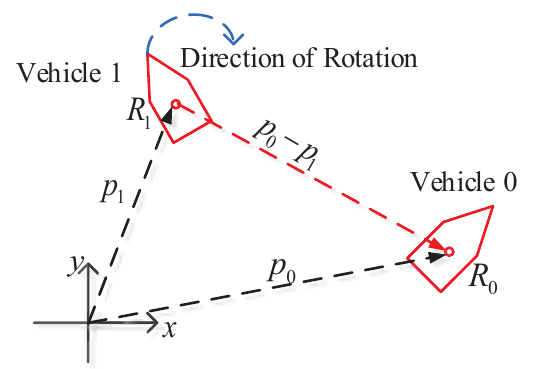}}
		\subfigure[]
		{\label{frame}\includegraphics[width=0.45\linewidth]{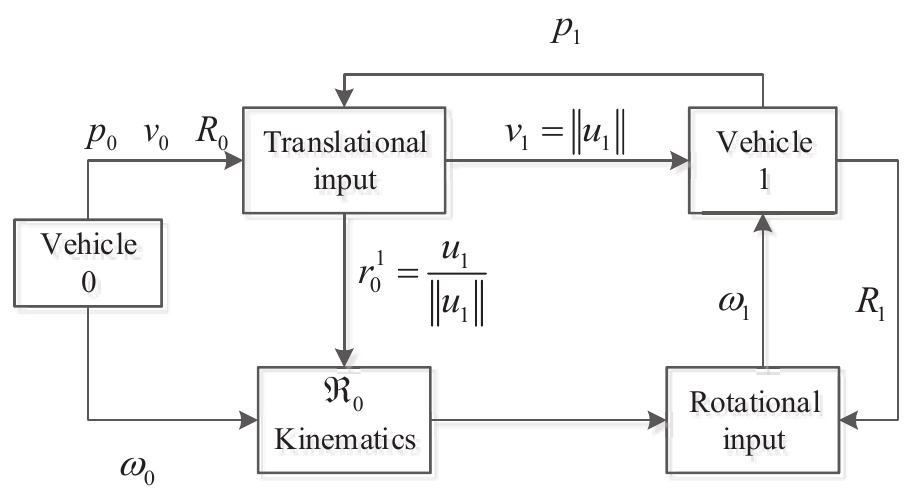}}
				
	}
	\caption{The tracking control strategy: (a). forward motion control  in a tracking control scenario; (b). the proposed two-stage framework.}
	\label{twostage}
\end{figure}

\subsection{Stability analysis}

In this subsection we present detailed analysis on the stability of vehicle $1$ with the control inputs $v_{1}$ and $\omega_{1}$.

From  \eqref{trackvi} and \eqref{R01}, the control vector $u_1$ can be equivalently stated as  $u_{1}\triangleq v_{1}\mathcal{R}_{0}e_{1}$.
Then, the kinematics $(\ref{modelp})$ can be rewritten as
\begin{equation}
\dot{p}_{1}=u_{1}+\Delta_{1}
\end{equation}
where $\Delta_{1}=v_{1}R_{1}e_{1}-v_{1}\mathcal{R}_{0}e_{1} \in \mathbb{R}^{2}$ can be seen as a perturbed term. Hence, the dynamics of the  position error $p_{01}$ and the attitude error $R_{01}$ are given by
\begin{equation} \label{modelde}
\begin{cases}
 \dot{p}_{01}= -k_{1}\sigma(p_{01})+\Delta_{1}  \\
 \dot{R}_{01}  =  -k_{2}R_{01}\sigma(R_{01}-R_{01}^{T})
\end{cases}
\end{equation}
Without the perturbation term $\Delta_{1}$, the closed dynamics $(\ref{modelde})$ can be written as
\begin{equation} \label{modele}
\begin{cases}
 \dot{p}_{01}= -k_{1}\sigma(p_{01})  \\
 \dot{R}_{01}  =  -k_{2}R_{01}\sigma(R_{01}-R_{01}^{T})
\end{cases}
\end{equation}

Let $\theta_{01}=\theta_{1}-\phi_{0}$ denote the Euler-angle error, where $\theta_{1}\in (-\pi,\pi]$ and $\phi_{0}\in (-\pi,\pi]$ are the Euler-angles corresponding to the rotation matrices $R_{1}$ and $\mathcal{R}_{0}$, respectively. Before giving  the main result of forward motion control for  leader trajectory tracking, the following lemmas are studied firstly.
\begin{lemma} \label{condition}
(see \cite{ding2016,zouyao}) Suppose the control inputs $v_{1}$ and $\omega_{1}$ stabilize the states of the (unperturbed) closed-loop error systems $(\ref{modele})$ asymptotically and there exists a bounded positive constant $\varphi$ such that $\|\Delta_{1}\|<\varphi \|\theta_{01}\|$.  Then the closed-loop error systems $(\ref{modelde})$ are asymptotically stable.
\end{lemma}

\begin{lemma} \label{condition2}
Under  Assumption $\ref{desiredw}$ and the saturated input
$(\ref{trackvi})$, the perturbation term $\Delta_{1}$ is bounded.
\end{lemma}
\begin{proof}
From Assumption $\ref{desiredw}$ and the properties of the saturation function, it holds that $v_{1} \leq\eta$, where $\eta$ is a bounded positive constant.  On the other hand, one has
$\|(\mathcal{R}_{0}-R_{1})e_{1}\| = \|(R_{01}-I_{2})e_{1}\| =2\sin(\theta_{01}/2)$.
Thus, we can obtain that
$
\|\Delta_{1}\| \leq v_{1}  \|(\mathcal{R}_{0}-R_{1})e_{1}\| <\varphi \|\theta_{01}\|
$
\end{proof}

The main result of the forward motion control  and trajectory tracking is given as follows.

\begin{thm} \label{thm1}
Suppose that the virtual control vector $u_{1}$, initial rotation matrix $R_{01}$ are required as $\|u_{1}\|\neq0$, $\mathrm{tr}(R_{01}(0))\neq -2$, and  Assumption $\ref{desiredw}$ is satisfied. Under the saturated inputs
$(\ref{trackvi})$ and $(\ref{trackwi})$, the non-holonomic vehicle 1 is able to track the leader 0  asymptotically.
\end{thm}

\begin{proof}
Under the assumption $\|u_{1}\|\neq 0$, the  intermediate rotation matrix $\mathcal{R}_{0}$ will  always be smooth.
Define the positive definite Lyapunov function as $
V=\frac{1}{2}p_{01}^{T}p_{01}+\mathrm{tr}(I_{2}-R_{01})
$.
Then, its time derivative  along the trajectories $(\ref{modele})$ is given by
\begin{equation} \label{dV}
\dot{V}=-k_{1}p_{01}^{T}\sigma(p_{01})-k_{2}((R_{01}-R_{01}^{T})^{\vee})\sigma((R_{01}-R_{01}^{T})^{\vee}) \leq 0
\end{equation}
where the fact that $\mathrm{tr}(R\hat{\omega})=-\omega(R-R^{T})^{\vee}$ is used. Note that from $\dot{V}\leq 0$, one has that  the closed-loop error system  $(\ref{modele})$ is Lyapunov stable. Note that $V(t)\leq V(0)$, which implies  the undesired equilibrium $\mathrm{tr}(R_{01})=-2$ is excluded. Based on the properties of saturation functions that $x\sigma(x)=0$ if and only if $x=0$, one has that $p_{01}$ and $R_{01}$ converge to the set on which $\dot{V}=0$ which is denoted by $S=\{(p_{01},R_{01})=(0,I_{2}) | \dot{V}=0\}$.  Thus, the closed-loop error system  $(\ref{modele})$ is asymptotically stable, i.e., $p_{01}\rightarrow 0$, $R_{01}\rightarrow I_{2}$ as $t\rightarrow \infty$.

In the light  of  Lemmas $\ref{condition}$ and $\ref{condition2}$, the asymptotic convergence of the closed-loop error system (\ref{modelde}) is proved.
While the position error $p_{01}\rightarrow 0$ as $t\rightarrow \infty$, the vector $r_{0}^{1} \rightarrow R_{0}e_{1}$ as $t\rightarrow \infty$. Since the rotation matrix $R \in SO(2)$ is a one parameter subgroup, one has that $\mathcal{R}_{0}^{T}R_{0}=I_{2}$ as  $t\rightarrow \infty$.

\end{proof}

\begin{rmk}
In Theorem {\ref{thm1}}, we assume $\|u_{1}\|\neq0$ all the time.  We note that the smoothness of the intermediate rotation matrix $\mathcal{R}_{0}$ is not guaranteed if $\|u_{1}\|=0$. To avoid approaching the point $\|u_{1}\|=0$, we can adopt the switching-based approach proposed in \cite{ding2016}. Once the vehicle is approaching and encounters $\|u_{1}\|=0$, we can  hold $v_{1}>0$ and $\mathcal{R}_{0}$. As long as the leader keeps moving, the state $\|u_{1}\|=0$ will not last all the time. After this moment, the controller will make the position error and attitude error converge to zero. Alternatively, $\|u_{1}\|=0$ can be avoided by choosing an appropriate gain $k_{1}$. For example, from $(\ref{trackvi})$, one has $v_{1} \geq \|v_{0}R_{0}e_{1}\|-\|-k_{1}\sigma(p_{01})\| \geq v_{0}-k_{1} > v_{min_{0}}-k_{1}$, where $v_{min_{0}}< v_{0}$ for all times. By choosing $v_{min_{0}}-k_{1}>0$, $\|u_{1}\|=0$ can never occur.
\end{rmk}

\section{Mobile  formation coordination of multiple non-holonomic vehicles} \label{sec:formation}

In this section, we study  mobile formation coordination  for multiple  non-holonomic vehicles. To be specific, we shall explore some properties of mobile formation with motion constraints   and then design a control law  to form and maintain a desired mobile formation with strict rigid-body motion. In the following, rigid-body motion (composed by   translational motion and rotational motion) is meant for a formation shape maintained by multiple mobile vehicle.

\subsection{Motion behaviors of mobile formation with motion constraints}

In the context of non-holonomic vehicle formation,  we give the following definitions on mobile formation with weak/strict rigid-body motion.


\begin{definition} \label{def:weak_fixed_formation}
A mobile formation with  \textit{weak} rigid-body motion  is a formation that  the relative positions of each mobile vehicle with respect to  one common mobile vehicle (in its body-fixed frame) are kept constant.
\end{definition}

\begin{definition} \label{def:fixed_formation}
A mobile formation with  \textit{strict} rigid-body motion  is a formation that the   relative positions between any two mobile vehicles (in their body-fixed frames) are kept constant.
\end{definition}

\begin{rmk}
For a more intuitive  understanding of the above definitions, three types of mobile formations involving  three vehicles are shown in Fig.~$\ref{formation}$.  Note that all the mobile formations shown in Fig.~$\ref{formation}$ involve rigid and fixed formation shapes \cite{oh2015,Anderson2008}, but the vehicle  motions of formation shapes are different.
The  mobile formation in Fig.~$\ref{formation}$(a) only has translational motion (i.e., the fixed formation shape  only admits a translational motion while all vehicles have a synchronized heading), which has been studied in \cite{yuxiao,liutengfei1,liutengfei2,miaozhiqiang}. Translational motion for the whole formation shape with synchronized vehicles' headings has limited motion freedoms, and is not the focus of this paper. In contrast,    the mobile formation with weak rigid-body motion in Fig.~$\ref{formation}$(b) has  both translational and rotational motion, and  the mobile formation with strict rigid-body motion in Fig.~$\ref{formation}$(c) can be seen as rotating a single rigid body.
According to Definition~\ref{def:weak_fixed_formation}, the mobile formation under a weak rigid-body motion only preserves  fixed relative positions of each vehicle with respect to one common vehicle, and the relative positions and attitudes between any two vehicles (except for the common vehicle) do not have to be constant. In contrast, according to   Definition~\ref{def:fixed_formation}, with a strict rigid-body motion, the mobile formation shape not only preserves fixed inter-vehicle distances  but also keeps constant relative attitudes  between any two vehicles.
\end{rmk}

The parallel formation and translational straight line formation, which are  two special cases of mobile formation with strict rigid-body motion as shown in Fig.~$\ref{formation}$ (d) and Fig.~$\ref{formation}$ (e), respectively,  are defined as follows.

\begin{definition}
A parallel formation is a formation where the headings of all vehicles are synchronized (or anti-synchronized), meanwhile the vehicle group keeps non-zero constant transverse offsets  and  zero longitudinal offsets   between any two vehicles.
\end{definition}

\begin{definition}
A translational straight line formation is  a formation where the angular speeds of all vehicles are zero meanwhile the vehicle group maintains constant transverse offsets and longitudinal offsets   between any two vehicles.
\end{definition}




\setlength{\abovecaptionskip}{0.cm}
\setlength{\belowcaptionskip}{-0.cm}
	\begin{figure}[!htb]
	\centering
	{		
		\subfigure[]
		{\label{nonformation}\includegraphics[width=0.35\linewidth]{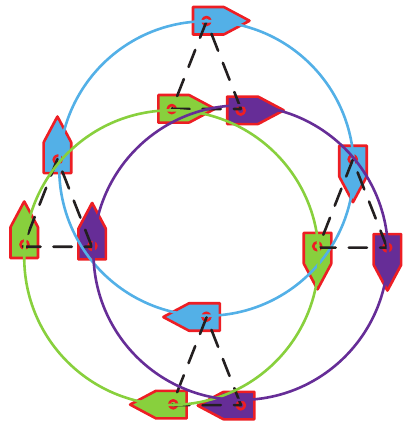}}
		\subfigure[]
		{\label{rigidformation}\includegraphics[width=0.45\linewidth]{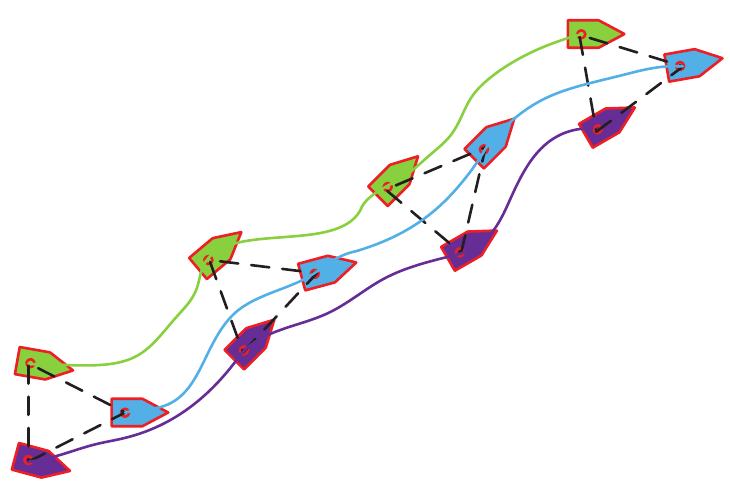}}
		\subfigure[]
		{\label{rigidformation}\includegraphics[width=0.35\linewidth]{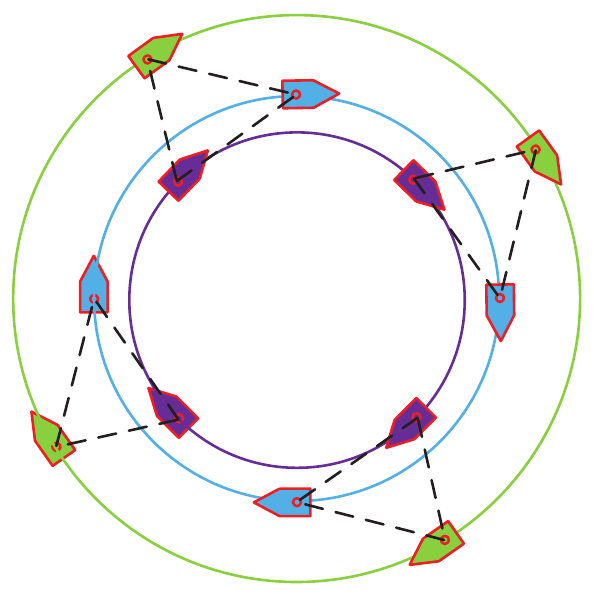}}
		\subfigure[]
		{\label{rigidformation}\includegraphics[width=0.45\linewidth]{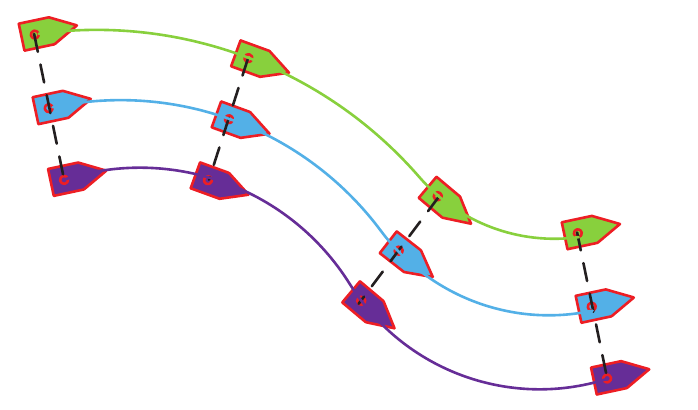}}
		\subfigure[]
		{\label{rigidformation}\includegraphics[width=0.45\linewidth]{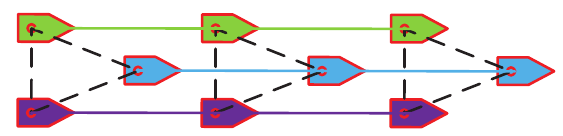}}
	}
	\caption{Illustrations  of mobile formations for non-holonomic vehicles 
: (a). a mobile   formation with only translational motion;  (b). a mobile   formation with weak rigid-body motion; (c). a mobile   formation with strict rigid-body motion; (d). a parallel  formation; (e) a translational straight line formation. (We refer the readers to the attached video for more motion demonstrations of these mobile formations.) }
	\label{formation}
\end{figure}


Since unicycle-type vehicles cannot move sideways, the weak/strict rigid-body motion for a mobile formation imposes further constraints which limit possible motion freedoms for multi-vehicle groups.  In particular, we consider a leader-follower framework to achieve a mobile formation. However, we remark that the vehicle indexed as $0$ is for control inputs that guide a mobile formation,  which is not necessarily a physical leader vehicle. In other words, the motion analysis for the defined mobile formation is applicable  to  both leader-follower structure and leaderless structure.  In the following,    we focus on studying  motion behaviors for a group of vehicles moving in  a mobile  formation with a strict/weak rigid-body motion.  Due to the existence of non-holonomic dynamics, a mobile  formation is required to satisfy certain  conditions that respect all vehicles' motion constraints. We firstly introduce the following definition.

\begin{definition} \label{adjoint}
(see \cite{liuyongfang}) For a mobile formation with a desired shape,  the desired trajectory of vehicle $i$ is called vehicle $i$'s adjoint orbit.
\end{definition}

Let $g_{ji}=g_{j}^{-1}g_{i}$ denote the relative configuration of vehicle $i$ with respect to vehicle $j$,  which is  given by
\begin{align*}
g_{ji}=\left[
               \begin{matrix}
              R_{ji} & R_{j}^{T}(p_{i}-p_{j}) \\
               0 &  1
               \end{matrix}
               \right]
          =\left[
               \begin{matrix}
              \cos\theta_{ji} & -\sin\theta_{ji} & x_{ji}^{j} \\
               \sin\theta_{ji} & \cos\theta_{ji} & y_{ji}^{j} \\
               0 & 0 & 1
               \end{matrix}
               \right]
\end{align*}
where $\theta_{ji}=\theta_{i}-\theta_{j}$. Then, the relative position of vehicle $i$ with respect to vehicle $j$ is defined as ${p}_{ji}^{j}=R_{j}^{T}(p_{i}-p_{j})$, and  the distance between vehicle $j$ and vehicle $i$ is defined as $d_{ji}=\| R_{j}^{T}(p_{i}-p_{j}) \|=\| p_{i}-p_{j} \|$.
The relative configuration of vehicle $i$ with respect to vehicle $0$ is defined as $g_{0i}=g_{0}^{-1}g_{i}$. Motion properties in a mobile  formation  for a group of non-holonomic vehicles are summarized  in Propositions $\ref{pop1}$, $\ref{popp}$, $\ref{pop2}$ and $\ref{popweak}$.

\begin{proposition} \label{pop1}
Consider a desired relative configuration of vehicle $i$ with respect to vehicle $0$  denoted by $\bar{g}_{0i}\in SE(2)$  in $(\ref{g0i})$, where $\bar{x}_{0i}^{0}$ and $\bar{y}_{0i}^{0}$ are constants which depend on the formation task, and  $\bar{\theta}_{0i}$  satisfies the equation $(\ref{theta0i})$.
\begin{equation} \label{g0i}
\bar{g}_{0i}=\left[
               \begin{matrix}
              \cos\bar{\theta}_{0i} & -\sin\bar{\theta}_{0i} & \bar{x}_{0i}^{0} \\
               \sin\bar{\theta}_{0i} & \cos\bar{\theta}_{0i} & \bar{y}_{0i}^{0} \\
               0 & 0 & 1
               \end{matrix}
               \right]
\end{equation}
\begin{equation} \label{theta0i}
\bar{\theta}_{0i}=\mathrm{atan2}(\omega_{0}\bar{x}_{0i}^{0}, v_{0}-\omega_{0}\bar{y}_{0i}^{0})
\end{equation}
Then the following properties hold.

(i). The trajectory $\tilde{g}_{i}=g_{0}\bar{g}_{0i}$ describes the vehicle $i$'s adjoint orbit.

(ii). The adjoint orbit for vehicle $i$ satisfies the following kinematic equation
\begin{equation} \label{g0i2}
\dot{\tilde{g}}_{i}=\tilde{g}_{i}(\mathrm{Ad}_{\bar{g}_{0i}^{-1}}\hat{\xi}_{0}+\hat{\bar{\xi}}_{0i})
\end{equation}
where $(\mathrm{Ad}_{\bar{g}_{0i}^{-1}}\hat{\xi}_{0}+\hat{\bar{\xi}}_{0i})$ is vehicle $i$'s adjoint velocity, and $\bar{\xi}_{0i}=[\bar{\omega}_{0i},0,0]^{T}$   satisfies the following kinematic equation
\begin{equation} \label{g0iequation}
\dot{\bar{g}}_{0i}=\bar{g}_{0i}\hat{\bar{\xi}}_{0i}
\end{equation}
where the angular speed is $\bar{\omega}_{0i}=\dot{\bar{\theta}}_{0i}$, and the heading error between vehicle $i$ and vehicle $0$ satisfies  $(\ref{theta0i})$.

(iii). For the case $\bar{x}_{0i}^{0}\neq 0$ and $\omega_{0}\neq 0$, the heading error $\bar{\theta}_{0i}$ cannot be zero.

(iv). For the case $\bar{x}_{0i}^{0}= 0$ or $\omega_{0}= 0$, the heading error $\bar{\theta}_{0i}\in\{0,\pi\}$, which corresponds  to a parallel formation or a translational straight line formation.
\end{proposition}

\begin{proof}
Once a mobile formation shape is achieved as shown in  Fig.~$\ref{formationfig}$, the matrix $\bar{g}_{0i}$ represents the desired configuration of vehicle $i$ expressed in the body frame of vehicle $0$. Based on  Definition $\ref{adjoint}$, the trajectory $\tilde{g}_{i}=g_{0}\bar{g}_{0i}$ represents vehicle $i$'s adjoint orbit. This proves (i).

Denote $\xi_{0}=[\omega_{0}, v_{0},0]^{T}$ and $\xi_{i}=[\omega_{i},v_{i},0]^{T}$, then  the kinematics of the adjoint orbit satisfy $(\ref{g0i2})$.
The adjoint velocity $(\mathrm{Ad}_{\bar{g}_{0i}^{-1}}\hat{\xi}_{0}+\hat{\bar{\xi}}_{0i})$ is rewritten as
\begin{equation} \label{adspeed}
(\mathrm{Ad}_{g_{0i}^{-1}}\hat{\xi}_{0}+\hat{\bar{\xi}}_{0i})^{\vee}=\left[
               \begin{matrix}
              \omega_{0}+\bar{\omega}_{0i} \\
               (v_{0}-\omega_{0}\bar{y}_{0i}^{0})\cos\bar{\theta}_{0i}+ \omega_{0}\bar{x}_{0i}^{0}\sin\bar{\theta}_{0i} \\
               -(v_{0}-\omega_{0}\bar{y}_{0i}^{0})\sin\bar{\theta}_{0i}+ \omega_{0}\bar{x}_{0i}^{0}\cos\bar{\theta}_{0i}
               \end{matrix}
               \right]
\end{equation}
To ensure that the adjoint orbit  satisfies the non-holonomic constraint, the third component of the velocity is required  to be zero, i.e.,
\begin{equation} \label{conditionlai}
-(v_{0}-\omega_{0}\bar{y}_{0i}^{0})\sin\bar{\theta}_{0i}+ \omega_{0}\bar{x}_{0i}^{0}\cos\bar{\theta}_{0i}=0
\end{equation}
Therefore, the condition $(\ref{theta0i})$ is obtained, which proves (ii).

The condition $(\ref{theta0i})$  shows that the two vehicles cannot have the same headings (i.e., $\bar{\theta}_{0i}\neq 0$) when $\bar{x}_{0i}^{0}\neq 0$ and $\omega_{0}\neq 0$. To keep the relative position fixed, the heading cannot be arbitrarily specified as in the holonomic case {\cite{dong2013}}, while the relative heading $\bar{\theta}_{0i}$ is determined by the relative position and coordinating speeds for non-holonomic vehicles. This proves (iii).

In the cases of $\omega_{0}=0$ or $\bar{x}_{0i}^{0} = 0$, one has $\bar{\theta}_{0i}\in\{0,\pi\}$ from $(\ref{theta0i})$. This
implies that the two vehicles have synchronized or anti-synchronized heading which
corresponds to a parallel formation or a translational straight line formation. This proves (iv).
\end{proof}

\setlength{\abovecaptionskip}{0.cm}
\setlength{\belowcaptionskip}{-0.cm}
\begin{figure}[!htb]
\centering
\includegraphics[width=2.2in]{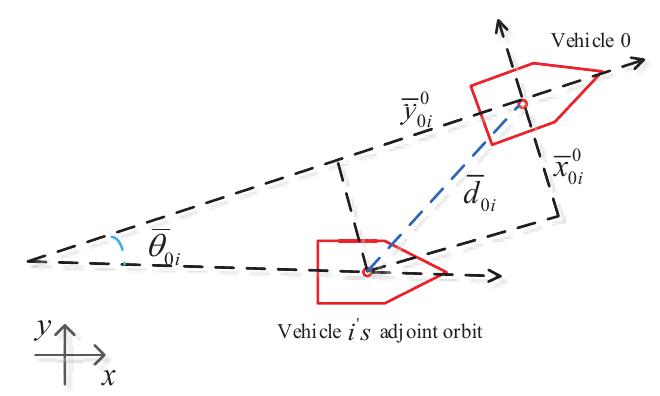}
\caption{Illustration of  relative configuration and adjoint orbit for two non-holonomic vehicles in a mobile formation. }
\label{formationfig}
\end{figure}


The relative configuration $\bar{g}_{0i}$ in $(\ref{g0i})$ $(\ref{theta0i})$ will be used to determine the adjoint orbit of vehicle $i$.  Regarding the condition $(\ref{theta0i})$, the following property holds.
\begin{proposition} \label{popp}
Suppose the vehicle $i$ moves along its adjoint orbit $\tilde{g}_{i}$ determined by $g_{0}\bar{g}_{0i}$.  Then the linear speed $v_{i}$ is positive if $v_{0}>0$.
\end{proposition}

\begin{proof}
Since the vehicle $i$ moves along its adjoint orbit, from the expression of adjoint velocity $(\ref{adspeed})$, one obtains
\begin{equation} \label{advi}
v_{i}=(v_{0}-\omega_{0}\bar{y}_{0i})\cos\bar{\theta}_{0i}+ \omega_{0}\bar{x}_{0i}\sin\bar{\theta}_{0i}
\end{equation}
From the condition  \eqref{conditionlai}, one has  $\frac{\sin\bar{\theta}_{0i}}{\cos\bar{\theta}_{0i}}=\frac{\omega_{0}\bar{x}_{0i}}{v_{0}-\omega_{0}\bar{y}_{0i}}$. To proceed, from $(\ref{advi})$, one obtains
\begin{align*}
\frac{v_{i}}{\cos\bar{\theta}_{0i}} &=(v_{0}-\omega_{0}\bar{y}_{0i})+ \omega_{0}\bar{x}_{0i}\frac{\sin\bar{\theta}_{0i}}{\cos\bar{\theta}_{0i}} \\
&=(v_{0}-\omega_{0}\bar{y}_{0i})+ \omega_{0}\bar{x}_{0i}\frac{\omega_{0}\bar{x}_{0i}}{(v_{0}-\omega_{0}\bar{y}_{0i})} \\
&= \frac{(v_{0}-\omega_{0}\bar{y}_{0i})^2+(\omega_{0}\bar{x}_{0i})^2}{(v_{0}-\omega_{0}\bar{y}_{0i})}
\end{align*}
Thus,  the linear speed $(\ref{advi})$ could be rewritten as
\begin{align} \label{advizf}
v_{i} &=\frac{(v_{0}-\omega_{0}\bar{y}_{0i})^2+(\omega_{0}\bar{x}_{0i})^2}{(v_{0}-\omega_{0}\bar{y}_{0i})}\cos\bar{\theta}_{0i}
\end{align}
 Based on the definition of function $\mathrm{atan2}$, we consider four cases to discuss the sign of $v_{i}$.

(i) In the first quartile: $\omega_{0}\bar{x}_{0i}>0$ and $v_{0}-\omega_{0}\bar{y}_{0i}>0$;

(ii) In the second quartile: $\omega_{0}\bar{x}_{0i}>0$ and $v_{0}-\omega_{0}\bar{y}_{0i}<0$;


(iii) In the fourth quartile:  $\omega_{0}\bar{x}_{0i}<0$ and $v_{0}-\omega_{0}\bar{y}_{0i}>0$.

(iv) In the coordinate axis except for the origin coordinate: $\omega_{0}\bar{x}_{0i}=0$ and $v_{0}-\omega_{0}\bar{y}_{0i}\neq0$  (or  $\omega_{0}\bar{x}_{0i}\neq0$ and $v_{0}-\omega_{0}\bar{y}_{0i}=0$  ).

For the case of (i), one obtains $\cos\bar{\theta}_{0i}>0$, which implies that $v_{i}>0$ from equation $(\ref{advizf})$. The same analysis can be used for cases of (ii) and (iii). For the case of (iv), from $(\ref{advi})$, one has $v_{i}=v_{0}-\omega_{0}\bar{y}_{0i}>0$ if  $\omega_{0}\bar{x}_{0i}=0$ and $v_{0}-\omega_{0}\bar{y}_{0i}>0$ (since $\bar{\theta}_{0i}=0$), and other situations  (i.e. $\bar{\theta}_{0i}=\{ \pi/2, \pi, 3\pi/2 \}$) can be proceeded with the same analysis. This completes the proof.

\end{proof}

\begin{rmk}
According to the singularity of $\mathrm{atan2}(0,0)$,  the cases including $\bar{x}_{0i}=0,\bar{y}_{0i}=v_{0}/\omega_{0}$ (which can be avoided by the defined formation task) and $\omega_{0}=0,v_{0}=0$ (which represents a static leader) are not considered in this paper.  The relative configuration $\bar{g}_{0i}$ in $(\ref{g0i})$ and $(\ref{theta0i})$ not only  determines adjoint orbit of vehicle $i$, but also guarantees the forward motion (i.e. positive linear speed). Since the linear speed of adjoint velocity is positive, the mobile formation could be achieved for coordinating multiple  fixed-wing UAVs to ensure  they move forward along their adjoint orbits with the corresponding  adjoint velocities.
\end{rmk}

\begin{rmk}
 If $v_{0}<0$, then the linear speed $v_{i}$ of the adjoint orbit which is determined by $\bar{g}_{0i}$ in $(\ref{g0i})$ $(\ref{theta0i})$ is positive based on the same analysis of Proposition $\ref{popp}$. In practice, some applications may also demand  that $v_{i}<0$ for the case $v_{0}<0$. To tackle this problem, based on the same analysis of Proposition $\ref{popp}$, the relative attitude $\bar{\theta}_{0i}$ can be redefined as $\bar{\theta}_{0i}=\arctan(\frac{\omega_{0}\bar{x}_{0i}^{0}}{v_{0}-\omega_{0}\bar{y}_{0i}^{0}})$ for the case of $v_{0}-\omega_{0}\bar{y}_{0i}<0$, and  $\bar{\theta}_{0i}$ is redefined as $\bar{\theta}_{0i}=\arctan(\frac{\omega_{0}\bar{x}_{0i}^{0}}{v_{0}-\omega_{0}\bar{y}_{0i}^{0}})+\pi$ for the case of $v_{0}-\omega_{0}\bar{y}_{0i}>0$. In addition, there exists a singularity at $v_{0}-\omega_{0}\bar{y}_{0i}=0$ if we use $\bar{\theta}_{0i}=\arctan(\frac{\omega_{0}\bar{x}_{0i}^{0}}{v_{0}-\omega_{0}\bar{y}_{0i}^{0}})$, and this can be avoided by the defined formation task.
\end{rmk}


  Propositions $\ref{pop1}$ and $\ref{popp}$ present details and calculations about the non-holonomic vehicle $i$'s adjoint orbit. Based on the results of adjoint orbit, the mobile formation with weak/strict rigid-body motion will be discussed. For a mobile formation  system involving multiple vehicles, the properties of a  mobile  formation with strict rigid-body motion are presented firstly.

\begin{proposition} \label{pop2}
For a  networked mobile formation control system with multiple non-holonomic vehicles, suppose  each  vehicle $i$ moves along its adjoint orbit $\tilde{g}_{i}$ determined by $g_{0}\bar{g}_{0i}$. Then the following properties hold.

(i). Except for the parallel formation and translational straight line formation, a mobile formation with strict rigid-body motion  can be achieved if and only if the speed ratios $v_{i}/\omega_{i}$ ($i=0,1,\cdots,n$) for each individual vehicle  are constants.

(ii). For the case of  parallel formation (i.e., $\bar{x}_{0i}^{0}= 0$), the vehicle 0 which guides the mobile formation can move with any bounded $v_{0},\omega_{0}$.

(iii) For the case of  translational straight line formation (i.e., $\omega_{0}= 0$), the vehicle 0 which guides the mobile formation can move with any bounded $v_{0}$.

(iv). The adjoint orbit  $g_{0}\bar{g}_{0i}$ can be also expressed as $g_{j}\bar{g}_{ji}$, which implies that the vehicle $i$'s adjoint orbit is unique in the inertial frame $\mathcal{F}_{I}$.

(v). To  maintain a strict rigid-body motion in mobile formation, the linear speeds $v_{i}$ for each  individual vehicles are not identical except for the translational straight line formation.

\end{proposition}

\begin{proof}
When each vehicle $i$ moves along its adjoint orbit $\tilde{g}_{i}$, the configuration of vehicle $i$ is given by $g_{i}=\tilde{g}_{i}=g_{0}\bar{g}_{0i}$. Then, the relative configuration of  vehicle $j$ with respect to vehicle $i$ is given by
\begin{equation} \label{ggij}
\bar{g}_{ij}=g_{i}^{-1}g_{j}=(g_{0}\bar{g}_{0i})^{-1}g_{0}\bar{g}_{0j}=\bar{g}_{0i}^{-1}\bar{g}_{0j}
\end{equation}
where $\bar{R}_{ij}=\bar{R}_{0i}^{T}\bar{R}_{0j}$, $\bar{p}_{ij}^{i}=\bar{R}_{0i}^{T}(\bar{p}_{0j}^{0}-\bar{p}_{0i}^{0})$. Thus, the relative position $\bar{p}_{ij}^{i}$ for any two vehicles is kept constant if and only if the rotation matrix $\bar{R}_{0i}$ is constant, which implies that the angle $\bar{\theta}_{0i}$ is constant. From equation $(\ref{theta0i})$, one can show that the angle $\bar{\theta}_{0i}$ is constant if and only if   the   speed   ratio $v_{0}/\omega_{0}$ is constant except for the cases of parallel formation and translational straight line formation.  From the adjoint velocity $(\mathrm{Ad}_{\bar{g}_{0i}^{-1}}\hat{\xi}_{0}+\hat{\bar{\xi}}_{0i})$ in $(\ref{adspeed})$, one has that the  speed  ratios $v_{i}/\omega_{i}$  are constants.
Furthermore, the distance between  vehicle $i$ and vehicle $j$ is $\bar{d}_{ij}=\|\bar{p}_{ij}^{i}\|=\|\bar{p}_{0j}^{0}-\bar{p}_{0i}^{0}\|$ which implies $\bar{d}_{ij}$ is a constant.
 Based on Definition $\ref{def:fixed_formation}$, it proves (i).

For the case of $\bar{x}_{0i}^{0}= 0$, the heading error is $\bar{\theta}_{0i}=\{0,\pi\}$ no matter whether the speeds $v_{0}, \omega_{0}$ are time-varying or constants. From  $(\ref{adspeed})$, one concludes that $\omega_{i}=\omega_{0}$ and $v_{i}=v_{0}-\omega_{0}\bar{y}_{0i}^{0}$ for vehicle $i,i=1,\cdots,n$ and  the speed ratios $v_{i}/\omega_{i}$ for each vehicle ($i=0,1,\cdots,n$) do not have to be constants. This proves (ii).

For the case of $\omega_{0}= 0$, the heading error is $\bar{\theta}_{0i}=\{0,\pi\}$ no matter whether the speed $v_{0}$ is time-varying or constants. From $(\ref{adspeed})$, one concludes that $\omega_{i}=0$ and $v_{i}=v_{0}$ for vehicle $i,i=1,\cdots,n$. This proves (iii).

If the mobile rigid formation is achieved, one has $g_{j}=g_{0}\bar{g}_{0j}$. Therefore, there holds $g_{j}\bar{g}_{ji}=g_{0}\bar{g}_{0j}\bar{g}_{ij}^{-1}=g_{0}\bar{g}_{0i}$. This shows (iv).

To maintain a mobile   formation with strict rigid-body motion, the linear speed  of vehicle $i$ is $v_{i}= (v_{0}-\omega_{0}\bar{y}_{0i}^{0})\cos\bar{\theta}_{0i}+ \omega_{0}\bar{x}_{0i}^{0}\sin\bar{\theta}_{0i}$, according to the vehicle $i$'s adjoint velocity. For the translational straight line formation, one has that the  speeds of all vehicles are $v_{0}$ and $\omega_{0}=0$, which shows that the speeds of all individual vehicles are the same.  However, for all other mobile  formations with strict rigid-body motion, the linear speeds are not the same for each individual vehicle.  This proves (v).
\end{proof}

\begin{rmk}
In Proposition \ref{pop2}, expect for the cases of parallel formation and translational straight line formation, a mobile formation with strict rigid-body motion is achieved if and only if all vehicles move along a circular motion.
\end{rmk}

\begin{rmk}
The seminal paper {\cite{Morbiditac}} firstly proposed coordination controllers to regulate multiple non-holonomic vehicles in a  formation moving as a rigid body. In this paper, we have extended the results of {\cite{Morbiditac}} in the following aspects. Firstly,  compared with the sufficient condition proposed in {\cite{Morbiditac}}, in this paper by the defined adjoint orbit, the proposed condition for  strict rigid-body motion is necessary and sufficient. Specifically speaking, the angular speed in  {\cite{Morbiditac}} is assumed to be constant for achieving a target formation, but we demonstrate a weaker condition that the ratio of linear speed to angular speed is constant except for the cases of parallel formation and translational straight line formation. Secondly, the relative positions and attitude of all inter-vehicles are analyzed thoroughly instead of  relative positions and attitudes in one common leader's coordinate frame as in {\cite{Morbiditac}}.
In addition, the condition of mobile formation with weak rigid-body motion is also considered which will be discussed in the sequel.
\end{rmk}

Regarding  mobile formation coordination under weak rigid-body motion in Definition~\ref{def:weak_fixed_formation}, we give the following proposition. Without loss of generality, we denote vehicle 0 as the common vehicle.
\begin{proposition} \label{popweak}
For a  networked formation system with multiple non-holonomic vehicles, if  each  vehicle $i$ moves along its adjoint orbit $\tilde{g}_{i}$ determined by $g_{0}\bar{g}_{0i}$, then the following properties hold.

(i). The mobile formation with weak rigid-body motion can be achieved for any bounded speeds $v_{0},\omega_{0}$, where vehicle $0$  guides the mobile formation.

(ii).  The relative positions and attitudes between  any  two  vehicles (except for the common vehicle) do  not  have  to  be constant.

(iii). The inter-vehicle distances are constant, (i.e., the formation shape is rigid).

(iv). To maintain  a mobile formation with weak rigid-body motion, the speeds $v_{i},\omega_{i}$ for each individual vehicles are not identical except for the translational straight line formation.
\end{proposition}

\begin{proof}
When each vehicle $i$ moves along its adjoint orbit $\tilde{g}_{i}$,  the relative configuration of  vehicle $i$ with respect to vehicle $0$ is given in $(\ref{g0i})$. Then, the relative positions of each vehicle $i$ with respect to vehicle $0$ are constant vectors (i.e. $\bar{p}_{0i}^{0}=[\bar{x}_{0i}^{0},\bar{y}_{0i}^{0}]^{T}$). Based on  Definition~\ref{def:weak_fixed_formation}, it proves the statement (i).

For the time-varying $(\omega_{0}\bar{x}_{0i}^{0})/(v_{0}-\omega_{0}\bar{y}_{0i}^{0})$,  the relative attitude of each vehicle $i$ with respect to vehicle $0$  in $(\ref{theta0i})$ is time-varying. Then, based on  $(\ref{ggij})$, the relative position and attitude of vehicle $j$ with respect to vehicle $i$ are given as  $\bar{R}_{ij}=\bar{R}_{0i}^{T}\bar{R}_{0j}$, $\bar{p}_{ij}^{i}=\bar{R}_{0i}^{T}(\bar{p}_{0j}^{0}-\bar{p}_{0i}^{0})$. Thus, for the time-varying $\bar{\theta}_{0i}$, $\bar{R}_{ij}$ and $\bar{p}_{ij}^{i}$ are time-varying. Therefore, the statement (ii) holds.

The distance between any two vehicles is given as $\bar{d}_{ij}=\|\bar{p}_{ij}^{i}\|=\|\bar{R}_{0i}^{T}(\bar{p}_{0j}^{0}-\bar{p}_{0i}^{0})\|=\|\bar{p}_{0j}^{0}-\bar{p}_{0i}^{0}\|$ which implies that (iii) holds.

Based on the analysis of Proposition $\ref{pop2}$, the statement (iv) is obvious.
\end{proof}

\begin{rmk}
Different from maintaining the mobile formation with strict rigid-body motion that requires either circular motion, parallel formation or translational straight line formation, to maintain a mobile formation with  weak rigid-body motion, the vehicle $0$ which guides the mobile formation can move with any bounded $v_{0},\omega_{0}$ and all other vehicles have more degrees of  motion freedoms and can perform other formation motions.
\end{rmk}

\begin{rmk}
From the analysis of Propositions~\ref{pop1}, \ref{popp}, \ref{pop2} and \ref{popweak}, one can conclude that the headings for each individual vehicle are not identical except for the cases of parallel formation and translational straight line formation, and the linear speeds  for each individual vehicle are not identical except for the case of translational straight line formation.
\end{rmk}


\begin{rmk}
In \cite{das2002,liangxinwu,Consolini2008,liangxinwu2} and references therein, based on the leader-follower approach, the formation tasks for non-holonomic vehicles are often determined by the body-fixed frame of the leaders or followers.  However, since the desired relative position of each vehicle is determined by each preceding vehicle,  the rigid formation shape cannot be guaranteed by the proposed controls in those papers, and therefore the mobile formation  with weak/strict rigid-body motion cannot be obtained. For example, suppose three vehicles are connected by a directed tree graph (i.e. $0\rightarrow 1 \rightarrow2$), and the formation task is defined as $\lim_{t\rightarrow \infty }g_{0}^{-1}g_{1}=\bar{g}_{01}$, and $\lim_{t\rightarrow \infty }g_{1}^{-1}g_{2}=\bar{g}_{12}$. Then, the relative configuration of vehicle 2 with respect to vehicle 0 is given as $\lim_{t\rightarrow \infty }g_{0}^{-1}g_{2}=\bar{g}_{01}\bar{g}_{12}$. Thereby, the distance between vehicle 0 and 2 is time-varying in the limit if $\bar{\theta}_{01}$ is time-varying. Thus, the rigid formation shape cannot be achieved. In this paper,  Propositions \ref{pop1} and \ref{popp} demonstrate the adjoint orbit and its corresponding properties, from which  the necessary and sufficient condition of mobile formation with strict rigid-body motion  is derived in Proposition \ref{pop2}.  Proposition \ref{popweak} provides a way to determine a mobile formation with weak rigid-body motion.  As far as we know, this is the first time that such conditions and properties are proposed.
\end{rmk}

\begin{rmk}
In this subsection, we not only analyze the motion properties of some mobile formation maneuvers, but also provide velocity inputs in $(\ref{adspeed})$   to maintain the mobile formation with weak/strict rigid-body motion. Besides, to maintain a mobile formation with weak/strict rigid-body motion (except for the translational straight line formation), the velocities of all vehicles are non-identical, which is different from the mobile formation with only translational motion in Fig.~$\ref{formation}$ (a) which are studied in many previous papers \cite{yuxiao,liutengfei1,liutengfei2,miaozhiqiang}.
\end{rmk}

\begin{rmk} \label{graphtree}
Different from formation shape control based on graph rigidity theory \cite{Anderson2008} that uses inter-agent distances to specify a desired formation shape,
the relative configuration $g_{j}^{-1}g_{i}$ is used here to describe a desired  formation shape  of non-holonomic vehicle $i$ with respect to non-holonomic vehicle $j$ in a mobile formation. Once a formation task $\bar{g}_{0i}$ is determined, the  shape in a mobile formation  is  rigid.
 Based on the group operation on $SE(2)$ for describing a formation task,  a mobile  formation with  weak/strict  rigid-body motion can be achieved under different underlying graph topologies, e.g.,    directed graphs, undirected graphs, leader-follower structure, or leaderless structure. In the following subsection, we present a formation control scheme based on a directed tree graph.
\end{rmk}

\subsection{An example of distributed mobile formation control} \label{treegraph}

In Propositions $\ref{pop2}$ and $\ref{popweak}$, the weak/strict rigid-body motion for a mobile formation is determined by the task $\bar{g}_{0i}$. However, each vehicle requires real-time information of vehicle speeds of the vehicle $0$. For  distributed control of a networked multi-vehicle formation   system, it is desirable  that the vehicle 0's information is only available to a few vehicles, but not to all follower  vehicles. In this subsection, based on the leader-follower structure and the results in Section~\ref{sec:twoStage}, by assuming a  directed tree graph, a fully distributed control for achieving a mobile formation with strict rigid-body motion  is proposed.
In a directed tree graph, each follower has only  one parent node. First, we present the following proposition.
\begin{proposition} \label{pop3}
Consider $n+1$ non-holonomic vehicles labeled with $i=0,1,\cdots,n$, and interacted by a directed tree graph. Let  vehicle $j$ be the only parent node of vehicle $i$, and  the desired  formation shape is given by the fixed relative position vectors $\bar{p}_{ji}^{j}=[\bar{x}_{ji}^{j},\bar{y}_{ji}^{j}]^{T}$.  Then,  the following properties hold.

(i). Except for the cases of parallel formation or translational straight line formation, by designing  appropriate controllers, a mobile  formation with strict rigid-body motion in the fully distributed networked formation control system can be achieved if and only if the speed  ratio $v_{0}/\omega_{0}$ is constant.

(ii). The desired relative position of vehicle $i$ with respect to any vehicle  is fixed.

\end{proposition}

\begin{proof}
Since vehicle $i$ only obtains the information of its parent node (i.e. vehicle $j$), the networked formation control system is fully distributed. For any two vehicles $i$ and $k$, the desired relative configuration of vehicle $i$ with respect to vehicle $k$ can be calculated as $\bar{g}_{ik}=\bar{g}_{ji}\cdots \bar{g}_{hk} $, where $j$ is the parent node of $i$, $h$ is the parent node of $k$, while nodes $k$ and $i$ do not interact directly. Thus, the relative position $\bar{p}_{ik}^{i}$ for any two vehicles $i,k$ is constant if and only if the matrices $\bar{g}_{ji},\cdots, \bar{g}_{hk} $ are constant, i.e.,   $\bar{\theta}_{ji}, \cdots, \bar{\theta}_{hk}$ are constant.  From $\bar{\theta}_{ji}=\mathrm{atan2}( \omega_{j}\bar{x}_{ji}^{j}, v_{j}-\omega_{j}\bar{y}_{ji}^{j})$, it implies that  the  speed ratio $v_{0}/\omega_{0}$ should be constant. This proves (i). Since $  \bar{p}_{ik}^{i}$ is a constant vector,  the statement (ii) is proved.

\end{proof}


By assuming that the   speed  ratio $v_{0}/\omega_{0}$ is constant (i.e the leader moves along a circular motion),  we design  formation controllers for achieving a mobile formation with strict rigid-body motion as follows.  Let $\tilde{g}_{i}=g_{j}\bar{g}_{ji}$ be  vehicle $i$'s adjoint orbit. Then the kinematic equation of adjoint orbit for vehicle $i$ is given by
\begin{equation} \label{gij}
\dot{\tilde{g}}_{i}=\tilde{g}_{i}(\mathrm{Ad}_{\bar{g}_{ji}^{-1}}\hat{\xi}_{j})
\end{equation}
where $\mathrm{Ad}_{\bar{g}_{ji}^{-1}}\hat{\xi}_{j}$ is the vehicle $i$'s adjoint velocity which satisfies the non-holonomic constraint from the result of  Proposition $\ref{pop1}$.

Since the adjoint orbit is the desired trajectory when the mobile formation with strict rigid-body motion is achieved, the adjoint orbits can be viewed as virtual leaders which should be tracked by individual follower vehicles in the networked formation control system.
To this end, equation $(\ref{gij})$ can be rewritten as follows
\begin{subequations}  \label{modelpij}
\begin{align} \label{modelppij}
\dot{\tilde{p}}_{i}  =& \tilde{v}_{i}\tilde{R}_{i}e_{1} \\ \label{modelRij}
\dot{\tilde{R}}_{i}=& \tilde{R}_{i}\hat{\tilde{\omega}}_{i}
\end{align}
\end{subequations}
where $\tilde{v}_{i}=(v_{j}-\omega_{j}\bar{y}_{ji}^{j})\cos\bar{\theta}_{ji}+ \omega_{j}\bar{x}_{ji}^{j}\sin\bar{\theta}_{ji}$ and $\tilde{\omega}_{i}=\omega_{j}$.  Based on the trajectory tracking result of Section~\ref{sec:prob}, the  formation control laws are  designed as follows
\begin{equation}  \label{vi}
v_{i} =\|u_{i}\|
\end{equation}
\begin{equation}  \label{wi}
\omega_{i} =-k_{2}\sigma((R_{ij}-R_{ij}^{T})^{\vee})+{\varpi}_{j}
\end{equation}
where the virtual control input vector is  given by
\begin{equation}
u_{i}=-k_{1}\sigma(\tilde{p}_{ij})+v_{j}{R}_{j}e_{1}
 \end{equation}
in which $\tilde{p}_{ij}=\tilde{p}_{i}-p_{i}$, and $R_{ji}=\mathcal{R}_{j}^{T}R_{i}$ with the intermediate rotation matrix $\mathcal{R}_{j}$ constructed as
\begin{equation}
\mathcal{R}_{j}=[r_{j}^{1},r_{j}^{2}]\in SO(2)
\end{equation}
with the vectors defined by
\begin{equation*}
r_{j}^{1}=\frac{u_{i}}{\|u_{i}\|} \in \mathbb{S}^{1}, \quad r_{j}^{2}=\left[
               \begin{matrix}
               -r_{j}^{1}(2,1)\\
               r_{j}^{1}(1,1)
               \end{matrix}  \right]  \in \mathbb{S}^{1}
\end{equation*}

The main result in this subsection is in Theorem 2.

\begin{thm} \label{thm2}
Consider  $n+1$ non-holonomic vehicles interacted  by a directed tree graph. Assume  $\|u_{i}\|\neq0$, $\mathrm{tr}(R_{ji}(0))\neq -2$, and   Assumption $\ref{desiredw}$ is satisfied.  If the speed  ratio $v_{0}/\omega_{0}$ is constant, then the mobile  formation with strict rigid-body motion in Proposition $\ref{pop3}$ can be achieved under the saturated inputs $(\ref{vi})$ and $(\ref{wi})$.
\end{thm}

\begin{proof}
According to the trajectory tracking control approach presented in Theorem $\ref{thm1}$, one can conclude that the saturated inputs $(\ref{vi})$ and $(\ref{wi})$ can drive the vehicle $i$ to converge to the trajectory $(\ref{modelpij})$ (i.e., its adjoint orbit). Then, we can obtain that $g_{i}^{-1}\tilde{g}_{i}=g_{i}^{-1}g_{j}\bar{g}_{ji} \rightarrow I_{3}$ as $t\rightarrow \infty$, i.e. $\lim_{t\rightarrow \infty}g_{j}^{-1}g_{i}=\bar{g}_{ji}$. Based on the  analysis in Proposition $\ref{pop3}$, the networked   formation control system achieves  a mobile formation with desired strict rigid-body motion.
\end{proof}

\begin{rmk}
If the formation tasks are determined  as the parallel formation or the translational straight line formation, the saturated inputs $(\ref{vi})$ and $(\ref{wi})$ are also applicable. In addition, all vehicles can move along an arbitrary desired trajectory for the cases of parallel formation.
\end{rmk}

\section{Numerical  simulation and experiment}\label{sec:sml}

\subsection{Numerical  simulation}\label{sec:sml1}

In the simulation, we consider a group of autonomous ground vehicles consisting of two followers and one leader, and the underlying interaction graph for three vehicles is described  by a directed tree graph (i.e. $0\rightarrow 1 \rightarrow2$). The leader (vehicle 0)'s inputs are given by $v_{0}=0.06\mathrm{m}/\mathrm{s}$, and $\omega_{0}=0.05\mathrm{rad}/\mathrm{s}$, which are constants. The initial configurations of the three vehicles are $g_{0}=\{0,0,0\}$, $g_{1}=\{-\pi/4, 0,-0.2\}$ and $g_{2}=\{\pi/4, 0, 0.2\}$.
The control gains are chosen as $k_{1}=k_{2}=0.3$. The desired formation shape is given by relative position vectors  $\bar{p}_{01}^{0}=[-0.1,-0.1]^{T}$ and $\bar{p}_{12}^{1}=[0,0.2]^{T}$.  The evolutions of all vehicles' states are depicted in Fig.~\ref{constant}, which shows that the linear speeds of each  individual non-holonomic vehicles are different when they are tasked to  maintain a mobile  formation with strict rigid-body motion. Fig.~\ref{constant} also demonstrates the evolution of vehicles' headings and adjoint velocity as discussed in Propositions~\ref{pop1}, \ref{popp}, and \ref{pop2}. The positiveness of  linear speeds  verifies a guaranteed forwarding motion in the forward motion control  proposed in Section~\ref{sec:twoStage}. The formation trajectories
of three non-holonomic vehicles in the 2D space are plotted in Fig.~\ref{2d}. All figures show that the desired mobile formation is achieved and maintained by the proposed controller.

We refer the readers to the attached video for more simulations and movies on controlling and maintaining   mobile formations with weak/strict rigid-body motion.

	

\setlength{\abovecaptionskip}{0.cm}
\setlength{\belowcaptionskip}{-0.cm}
	\begin{figure}[!htb]
	\centering
	{		
		\subfigure[Relative positions]
		{\label{ed}\includegraphics[width=0.48\linewidth]{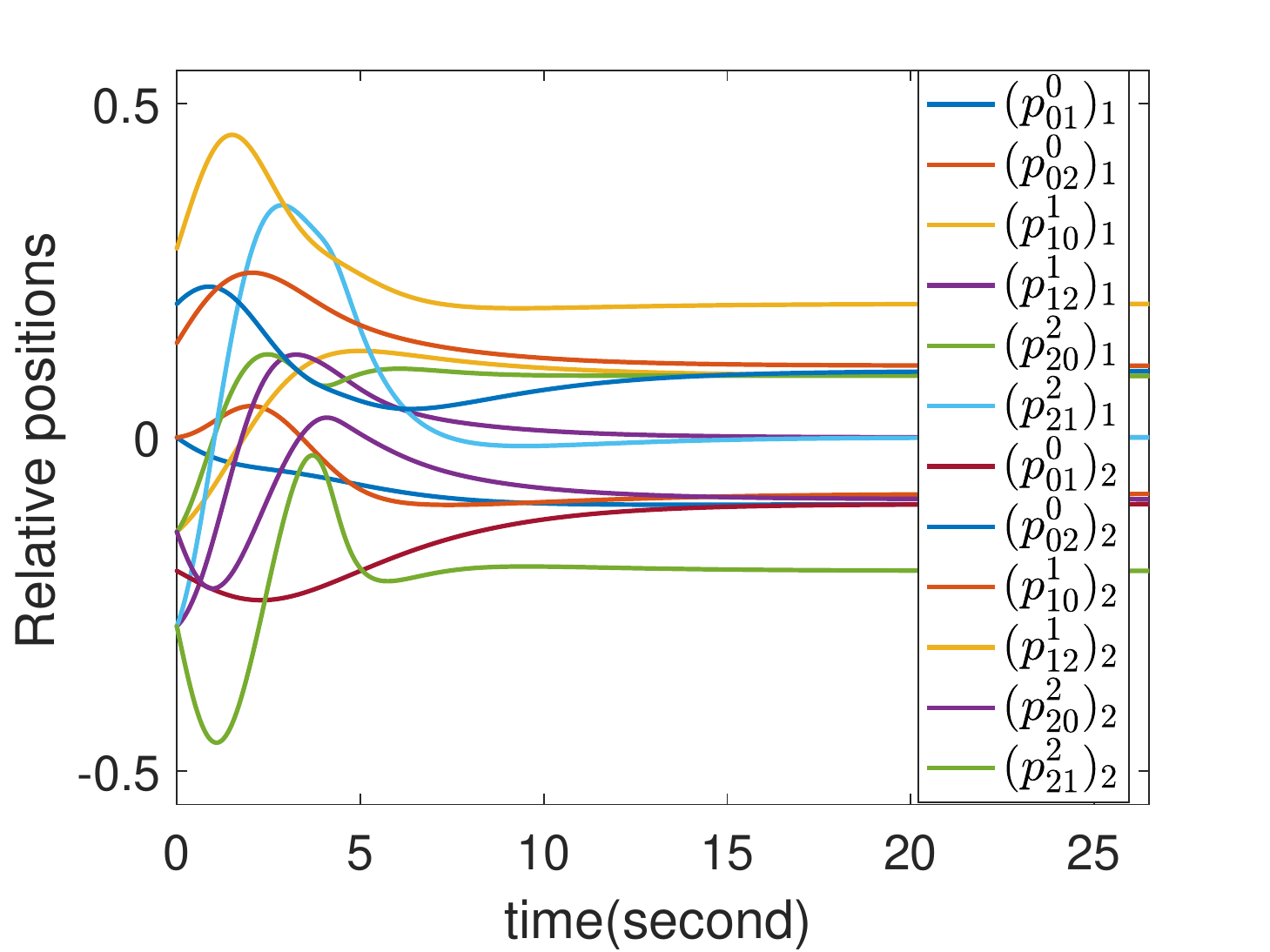}}
		\subfigure[Vehicles' headings]
		{\label{etheta}\includegraphics[width=0.48\linewidth]{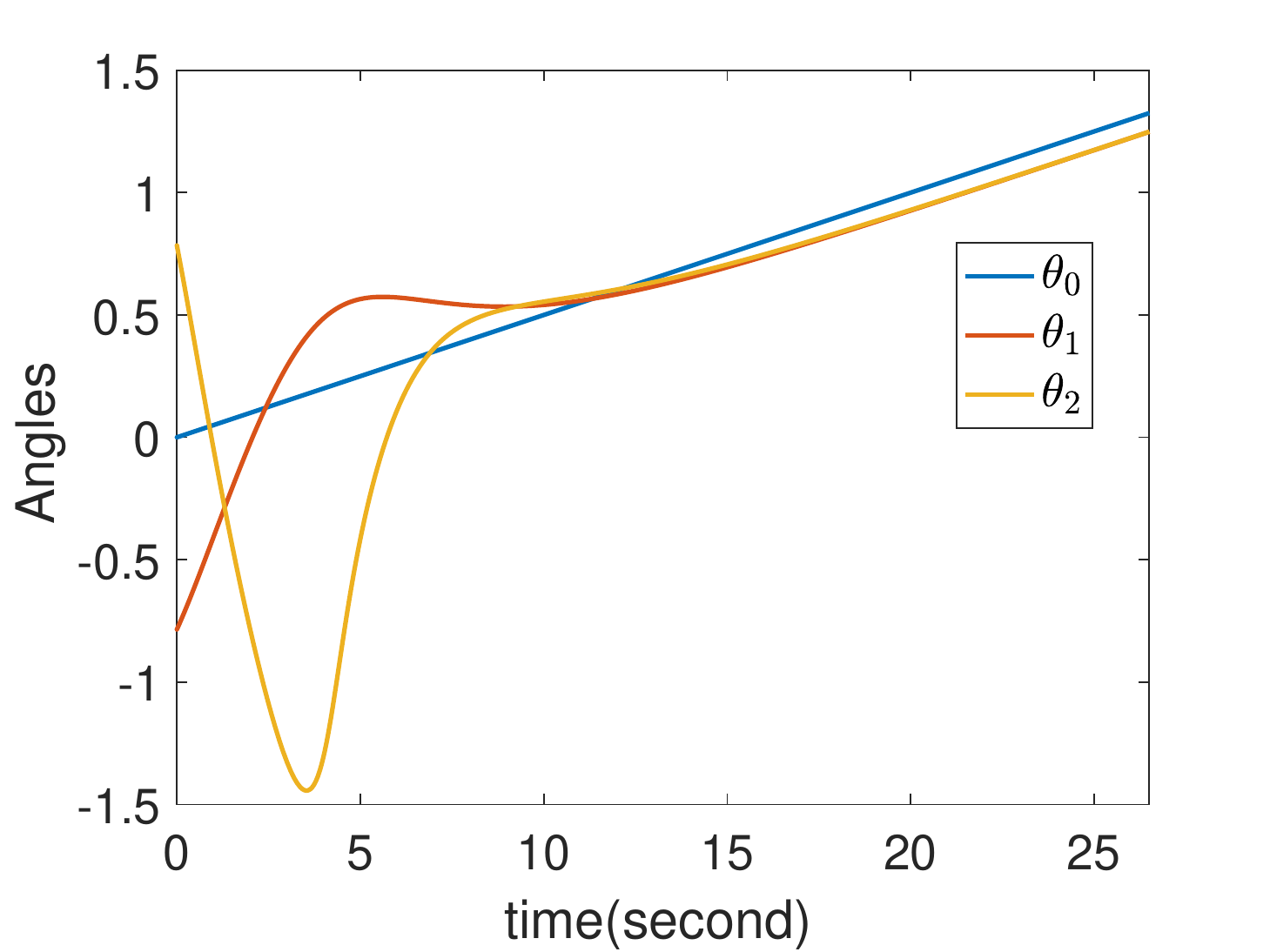}}
	
		\subfigure[Linear speeds]
		{\label{ev}\includegraphics[width=0.48\linewidth]{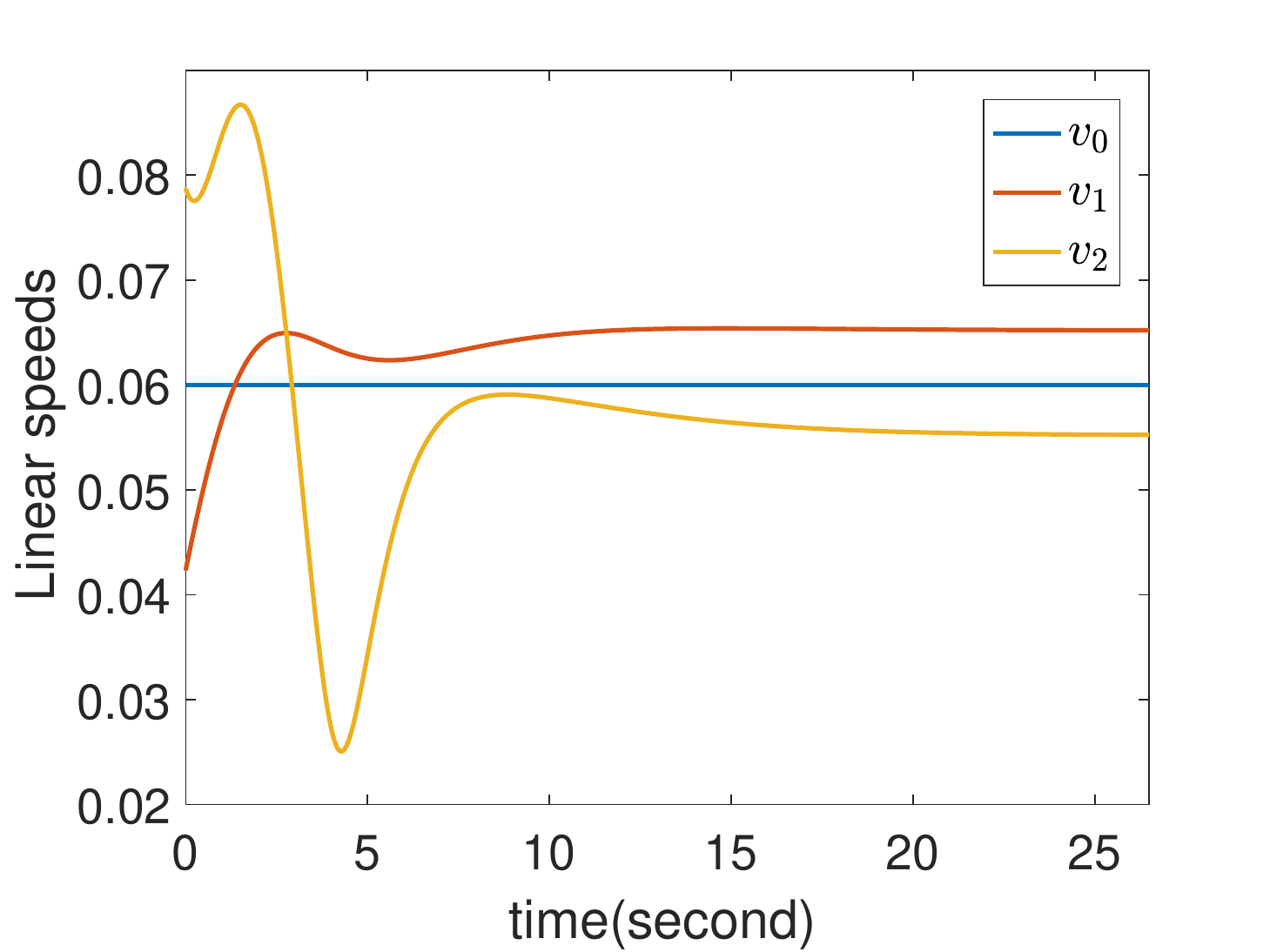}}	
		\subfigure[Angular speeds]
		{\label{ew}\includegraphics[width=0.48\linewidth]{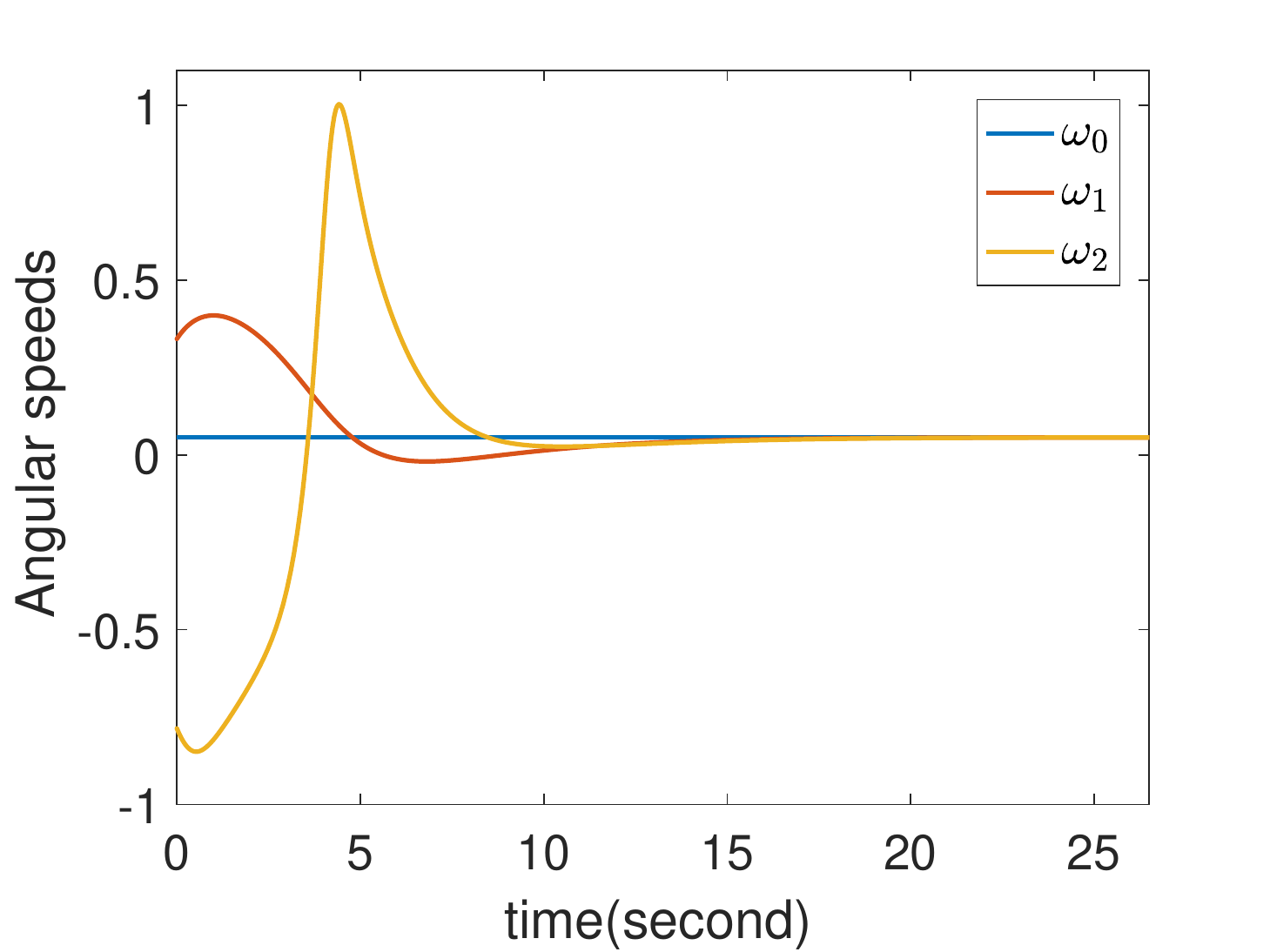}}				
	}
	\caption{The  evolutions of relative positions between any two vehicles, and each vehicle's headings, linear speeds and angular speeds  in the numerical simulation}
	\label{constant}
	\hspace{5cm}
\end{figure}

\setlength{\abovecaptionskip}{0.cm}
\setlength{\belowcaptionskip}{-0.cm}
\begin{figure}[!htb]
\centering
\includegraphics[width=3.0in]{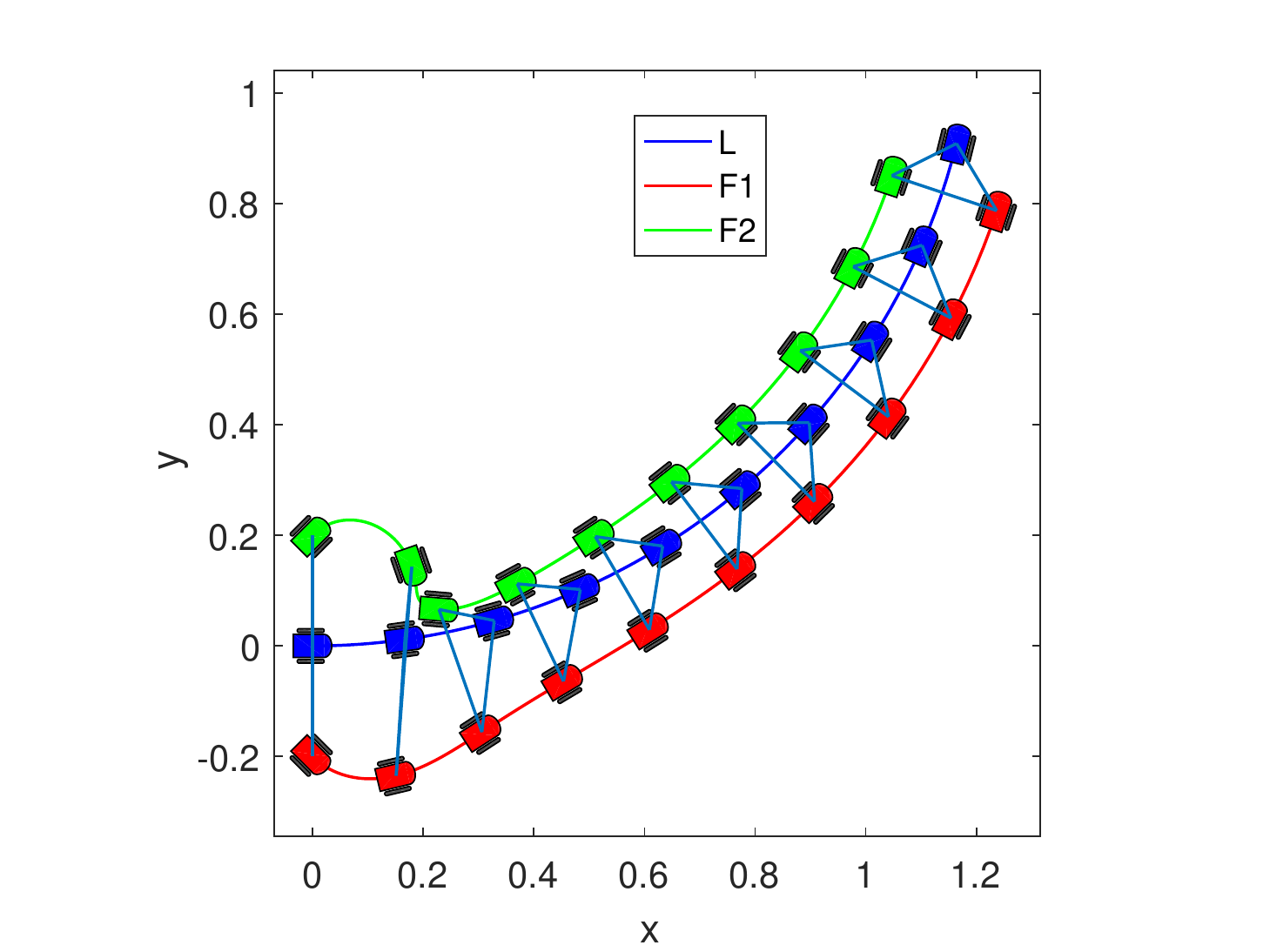}
\caption{The mobile formation with strict rigid-body motion   in the numerical simulation}
\label{2d}
\end{figure}

\subsection{Experiment}\label{sec:sml2}

To further demonstrate the applicability of the proposed scheme, real experiments on a physical multi-robotic system are performed. The experimental platform uses three non-holonomic robots named the wheeled E-puck robots {\cite{epuck,sunzhongqi}}.   In the experiment, we model the interaction among the three E-puck robots with a directed tree graph (i.e. $0\rightarrow 1 \rightarrow2$), as the same in  numerical simulations.  The robots move on a smooth and flat floor as shown in Fig.~\ref{e2d}.

The initial conditions, formation shape and gains are chosen the same as that in numerical simulation. The evolutions of relative positions, vehicles' headings, linear speeds and angular speeds are plotted in Fig.~\ref{estates} (plotted by Matlab with sampled data from the experiments). Fig.~\ref{e2d} shows the real-time trajectories of three robots in the mobile formation group, which are captured by a video camera. The experiments further validate both the forward motion control  proposed in Section~\ref{sec:twoStage} and mobile formation coordination under  strict rigid-body motion proposed in Section~\ref{sec:formation}.



	

\setlength{\abovecaptionskip}{0.cm}
\setlength{\belowcaptionskip}{-0.cm}
	\begin{figure}[!htb]
	\centering
	{		
		\subfigure[Relative positions]
		{\label{ed}\includegraphics[width=0.48\linewidth]{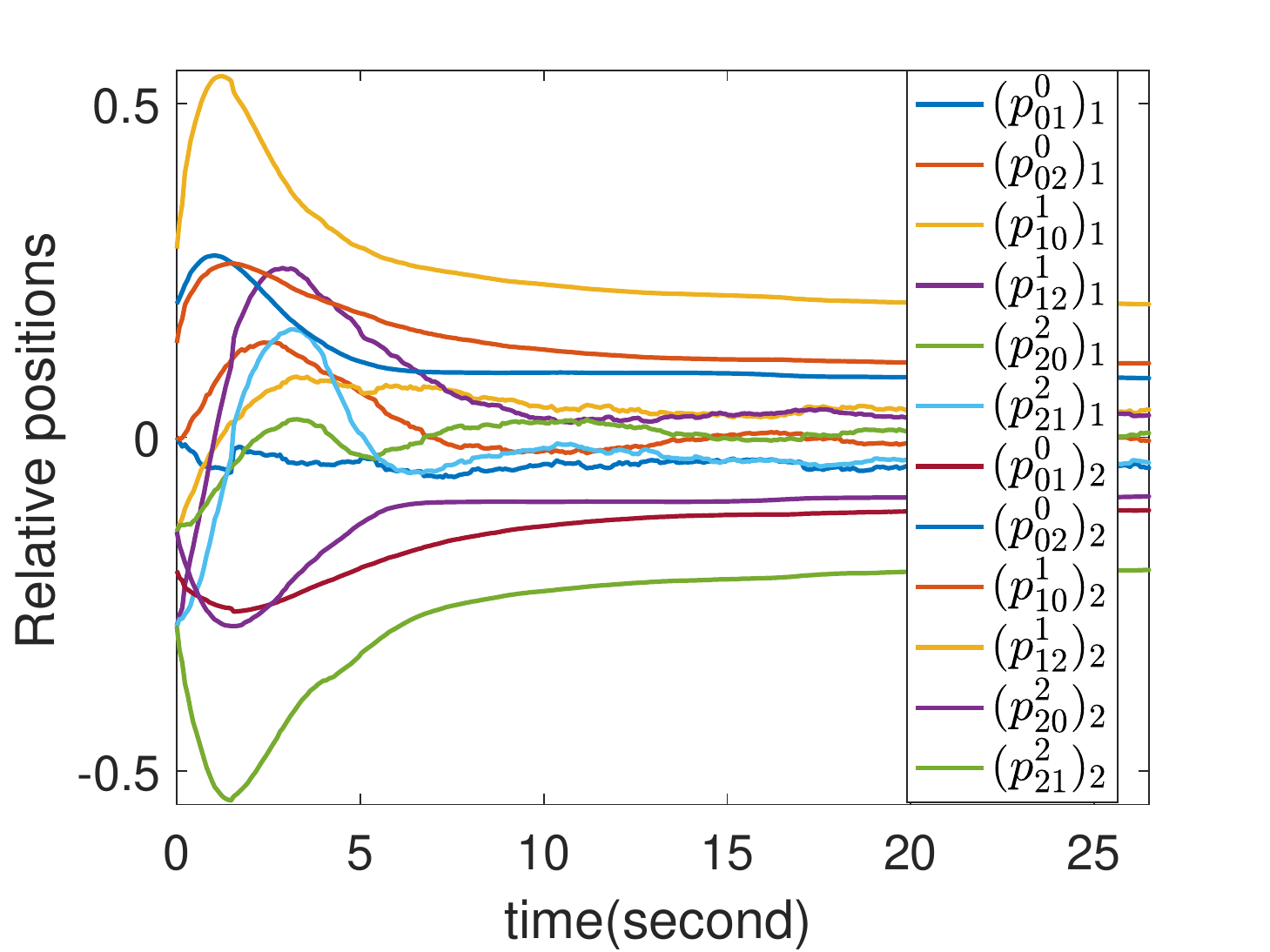}}
		\subfigure[Vehicles' headings]
		{\label{etheta}\includegraphics[width=0.48\linewidth]{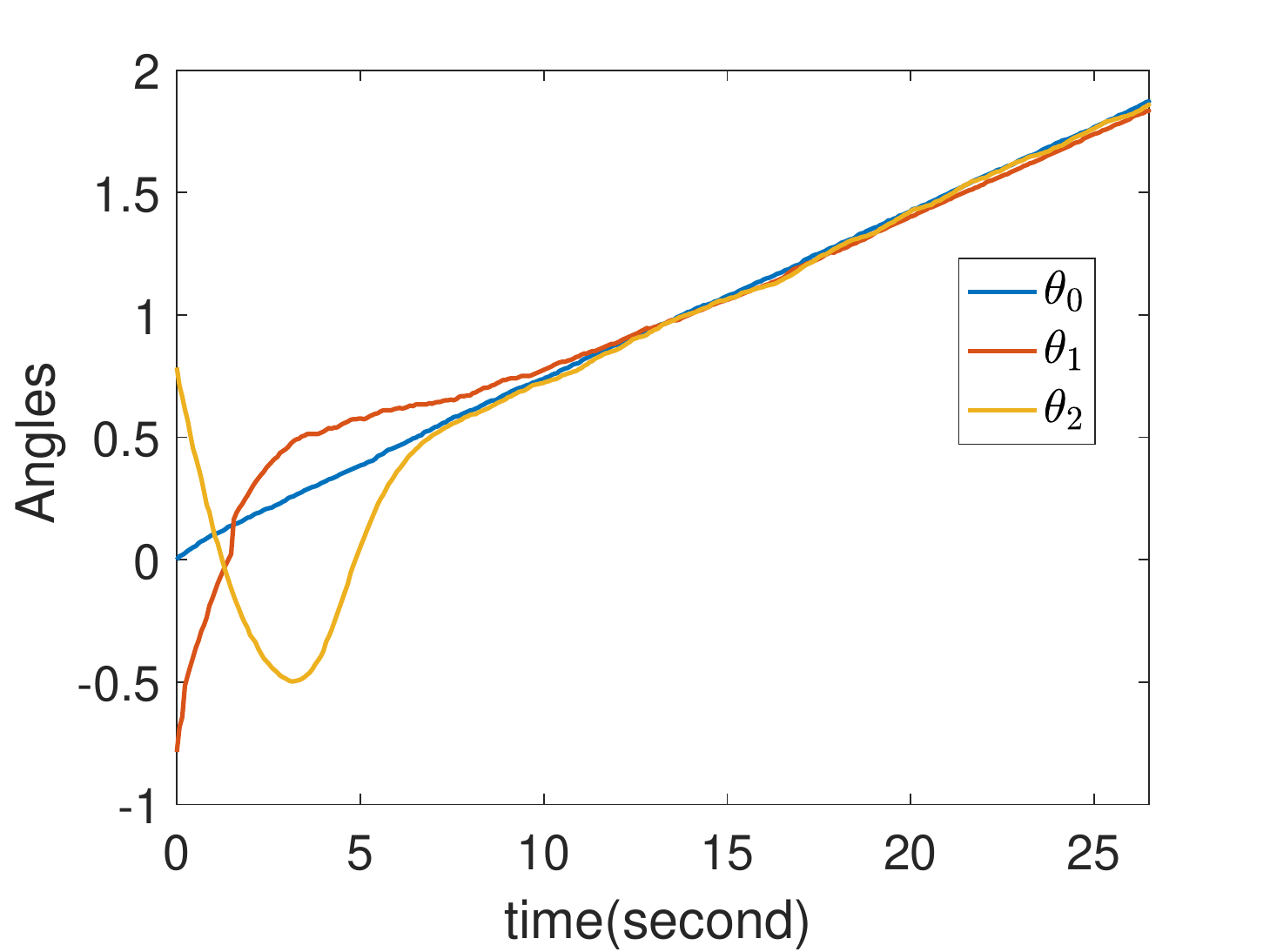}}
	
		\subfigure[Linear speeds]
		{\label{ev}\includegraphics[width=0.48\linewidth]{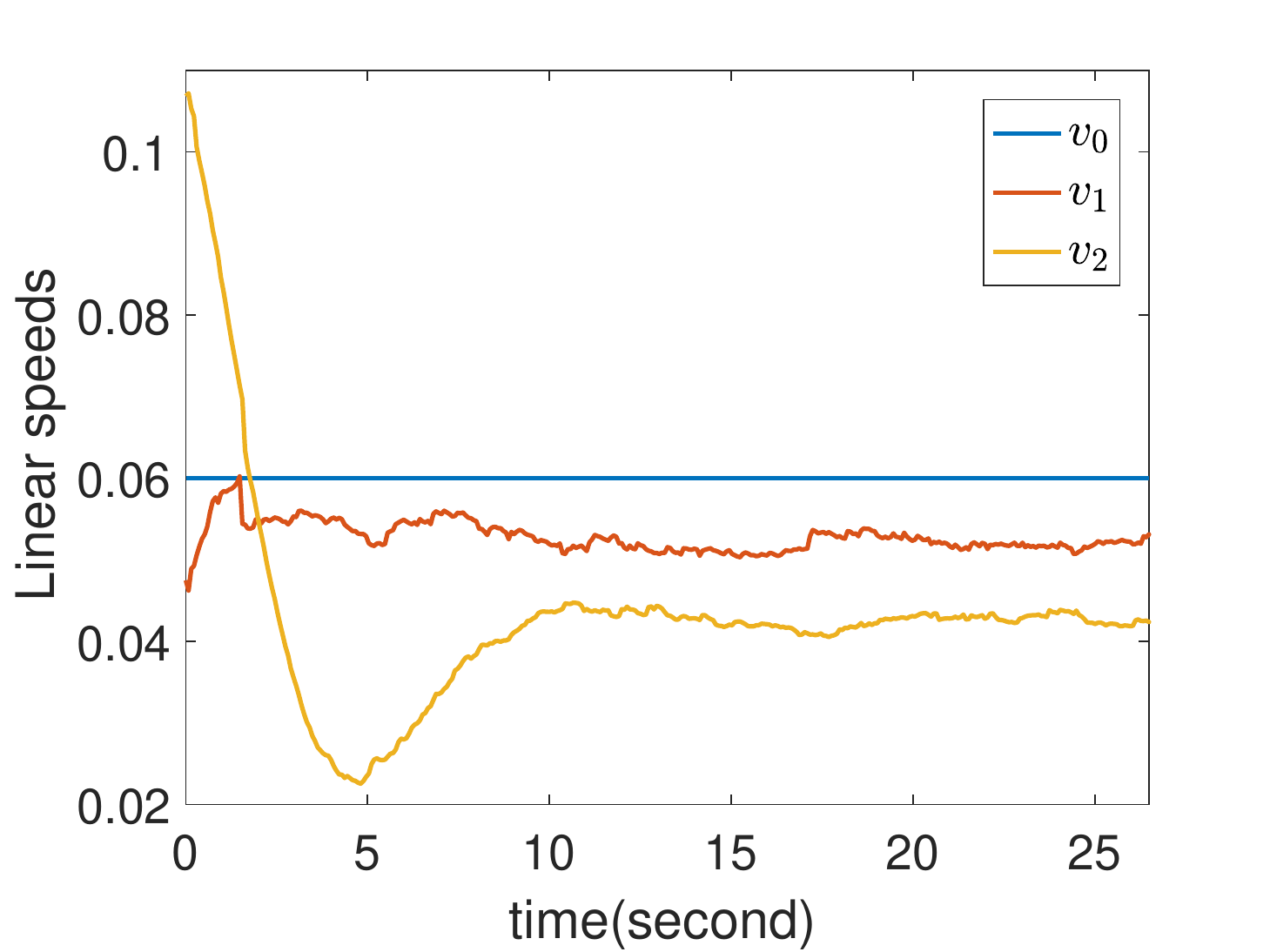}}	
		\subfigure[Angular speeds]
		{\label{ew}\includegraphics[width=0.48\linewidth]{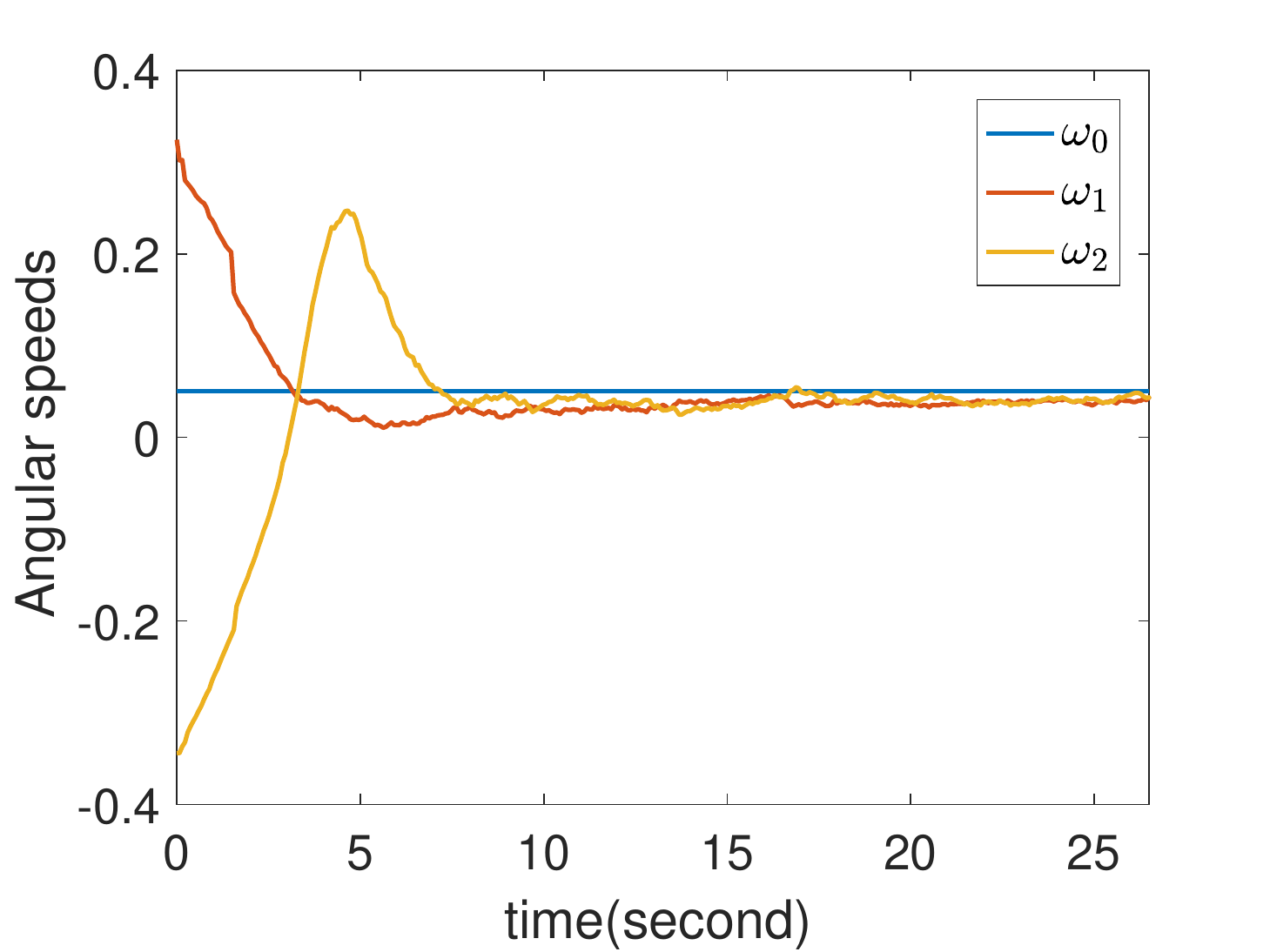}}				
	}
	\caption{The  evolutions of relative positions between any two vehicles, and each vehicle's headings, linear speeds and angular speeds  in real experiment}
	\label{estates}
	\hspace{5cm}
\end{figure}

%

\setlength{\abovecaptionskip}{0.cm}
\setlength{\belowcaptionskip}{-0.cm}
	\begin{figure}[!htb]
	\centering
	{		
		\subfigure[t=0s]
		{\label{t0}\includegraphics[width=0.48\linewidth]{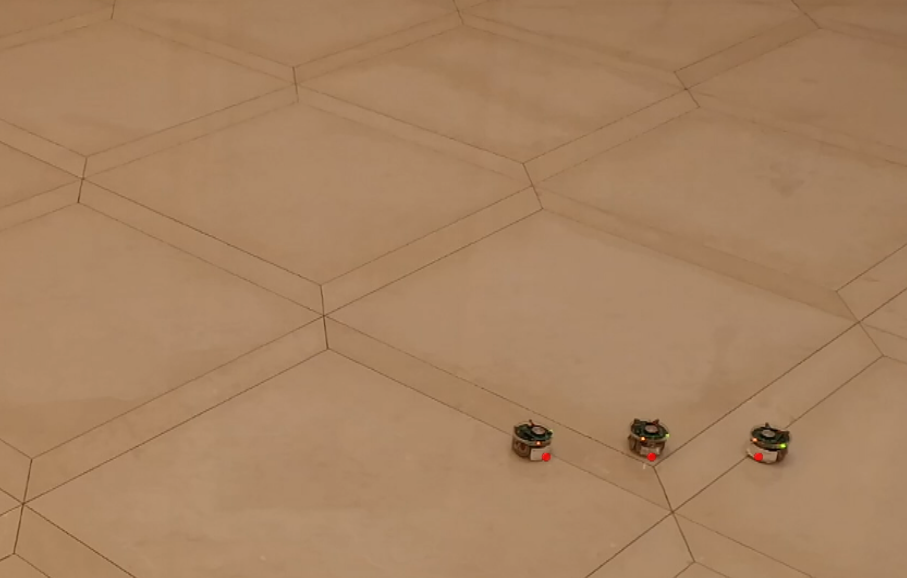}}
		\subfigure[t=5s]
		{\label{t5}\includegraphics[width=0.49\linewidth]{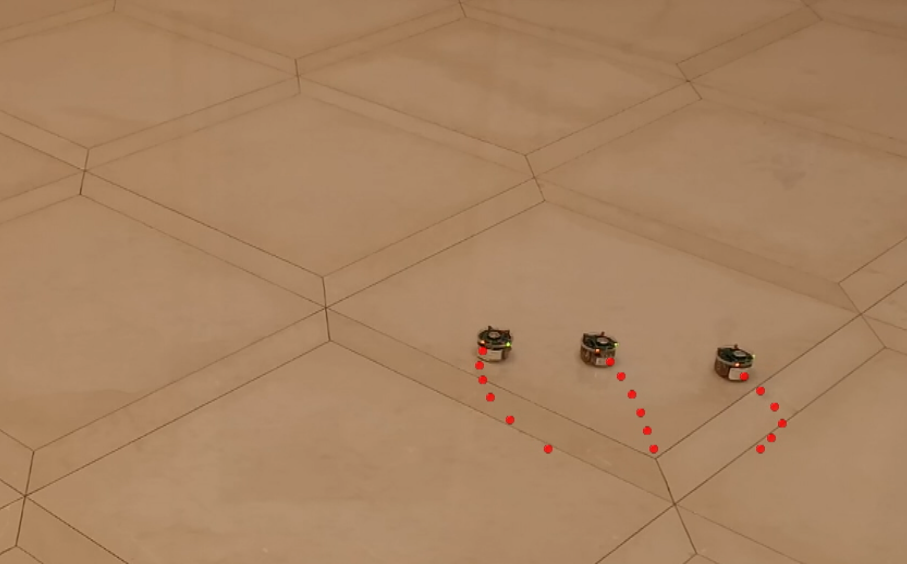}}
	
		\subfigure[t=13s]
		{\label{t13}\includegraphics[width=0.49\linewidth]{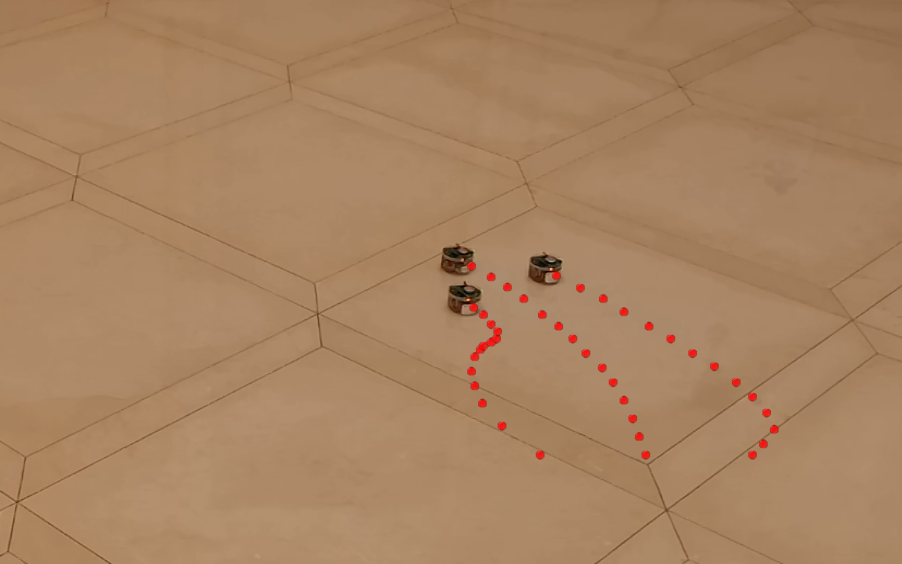}}	
		\subfigure[t=26s]
		{\label{t26}\includegraphics[width=0.49\linewidth]{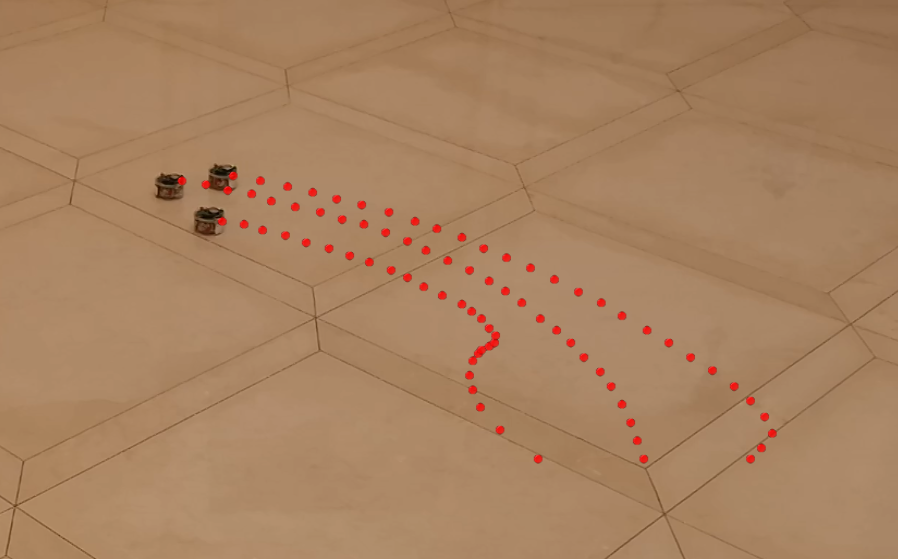}}				
	}
	\caption{Trajectories of three robots in a mobile formation experiment, with $t=0s,t=5s,t=13,t=26s$}
	\label{e2d}
	\hspace{5cm}
\end{figure}


\section{Conclusion}\label{sec:cls}

In this paper, the forward motion control for leader tracking and mobile formation with weak/strict rigid-body motion  for multiple non-holonomic vehicles are studied. An intermediate attitude which includes the relative position and attitude error to a leader vehicle is proposed, and the translational controller and rotational controller are designed in a two-stage framework. The formation behavior of   mobile  formations with  different rigid-body motion constraints for multiple non-holonomic vehicles is explored, and we demonstrate that the headings and linear speeds for individual vehicles are not identical when a group of non-holonomic  vehicles maintain a mobile  formation under rigid-body motion constraints. The behavior of the mobile  formation is analyzed  and a formation control strategy is proposed to achieve a mobile formation for multiple non-holonomic vehicles with strict rigid-body motion.
Both numerical simulations and experiments are provided to validate the performance of the proposed control approach.

\bibliographystyle{elsarticle-num}

\end{document}